\documentclass[12pt]{amsart}

\usepackage{amsfonts}
\usepackage{amsmath}
\usepackage{color}
\newtheorem{thm}{Theorem}[section]

\newtheorem{prop}[thm]{Proposition}
\newtheorem*{prob*}{Problem}
\newtheorem*{thm*}{Theorem}

\theoremstyle{definition}

\newtheorem*{defn*}{Definition}
\newtheorem{rem}[thm]{Remark}

\newtheorem*{rem*}{Remark}
\numberwithin{equation}{section}

\newcommand{\C}{\mathbb C}
\newcommand{\R}{\mathbb R}

\DeclareMathOperator{\const}{const}

\DeclareMathOperator{\Res}{Res}

\DeclareMathOperator{\E}{\mathbb{E}}

\DeclareMathOperator{\diag}{diag}

\newcommand{\Tr}{\mathop{\mathrm{Tr}}}

\begin{document}
\title[On characteristic polynomials]
 {\bf{On characteristic polynomials for a generalized chiral random matrix ensemble with a source.}}

\author{Yan Fyodorov}
\address{King's  College  London,  Department  of  Mathematics,  London  WC2R  2LS,  United  Kingdom}\email{yan.fyodorov@kcl.ac.uk}

\author{Jacek Grela}
\address{LPTMS, CNRS, Univ. Paris-Sud, Universit\'{e} Paris-Saclay, 91405 Orsay, France}
\address{M. Smoluchowski Institute of Physics and Mark Kac Complex Systems Research Centre, Jagiellonian University,  PL--30348 Krak\'ow, Poland}
\email{jacekgrela@gmail.com}

\author{Eugene Strahov}
\address{Department of Mathematics, The Hebrew University of
Jerusalem, Givat Ram, Jerusalem
91904, Israel}
\email{strahov@math.huji.ac.il}

\begin{abstract}

We evaluate averages involving characteristic polynomials, inverse characteristic polynomials and ratios of characteristic polynomials for a $N\times N$ random matrix taken from the $L$-deformed chiral Gaussian Unitary Ensemble with an external source $\Omega$.  Connection with a recently studied statistics of bi-orthogonal eigenvectors in the complex Ginibre ensemble, see arXiv:1710.04699,  is briefly discussed as a motivation to study asymptotics of these objects in the  case of external source proportional to the identity matrix. In such a case we retrieve the kernel related to Bessel/Macdonald functions.
\end{abstract}

\maketitle
\section{Introduction}

Random characteristic polynomials are fascinating objects admitting rich mathematical structures, e.g. explicit expressions for expected moments, products and ratios \cite{BrezinHikami,MehtaNormand2001,FyoStra2003a,StraFyo2003,Baik2003,AkVer2003,AP2004:GGCHARDET,BumpGamburd2006,BS2006:CHARDETS,Conrey2003,Kras2007,Conrey2008,KieburgGuhr2010}, duality formulas and integrability \cite{Des2009:DUALITYBETA,OsKanz2010,Gre2015:DIFFMETHOD} as well as relations with the kernel of the underlying determinantal processes \cite{BleherKuijlaars,Nikeghbali2017}.
Study of random characteristic polynomials is also strongly motivated by applications in Number Theory \cite{KS2000:RMTRIEMANNZEORES,HKOC2001,Conrey2008,FHK2012} as well as in many areas of Theoretical Physics ranging from Quantum Chaos $\&$ Quantum Chaotic Scattering \cite{AS1995:QCHAOS,FyoKhor99,FyoNock2015}, Anderson localization \cite{Shcherbinaloc} and theory of glassy systems \cite{landscapes}  to  QCD \cite{FA2003,SVQCD,BNW2013:BURGERSQCD} and string theory \cite{Kim2014:CHARDETBRANES}. In a recent work \cite{FyodorovEigenvectors} yet another application along these lines was revealed by the first author of the present paper, a relation between the distribution of the overlap between left and right eigenvectors of complex non-selfadjoint matrices from complex Ginibre ensemble and averages of inverse characteristic polynomials in a particular deformed version of the chiral Hermitian ensemble. In this paper we investigate the mathematical structure of the latter ensemble in more detail.

   Our main object of interest here is an $L$-deformed chiral Gaussian Unitary Ensemble with an external source specified via the following joint probability density function defined on the space of complex matrices $X$ of size $N\times N$:
\begin{equation}\label{distribution}
\mathcal{P}^{(L)}(X)dX=c\cdot e^{-\Tr\left(XX^*\right)+\Tr\left(\Omega X\right)+\Tr\left(X^*\Omega^*\right)}
\left[\det\left(X^*X\right)\right]^L\prod\limits_{i,j=1}^NdX_{i,j}^RdX_{i,j}^I,
\end{equation}
where $\Omega$ is a (in general, complex-valued) fixed non-random source matrix of size $N\times N$, $X_{i,j}=X_{i,j}^R+iX_{i,j}^{I}$ denote the sums of the real and imaginary parts of the matrix entries $X_{i,j}$, and $c$ is an appropriate normalization constant.
For $L=0,\Omega=0$,  Eq.(\ref{distribution}) defines the standard chiral Gaussian Unitary Ensemble (or, in a different interpretation, the Wishart-Laguerre ensemble) for which the general expectation values of product and ratios of characteristic polynomials were considered in detail in \cite{FA2003,FS2002:CHIRALGUECORR}. For $L=0,\Omega\ne 0$ the ensemble coincides with one studied in Ref. \cite{DesrosiersForrester}.
Note also that the ensemble defined via (\ref{distribution}) (but with a Hermitian source $\Omega=\Omega^*$) is actually a special limiting case of a version of the QCD-motivated ensemble considered in \cite{QCDfinT}. The latter paper addressed general $n-$point correlation functions of eigenvalue densities for such type of a chiral ensemble, but their knowledge is not sufficient for our goals of studying the averages involving  (inverse) characteristic polynomials associated with Eq.(\ref{distribution}).

The outline of the paper is as follows. In Section \ref{SectionDensity} we derive the joint pdf of the singular values for the matrix model given by equation \eqref{distribution}. Sections \ref{ICPSection} to \ref{KCPSection} develop compact integral representations of:
\begin{itemize}
\item averaged inverse characteristic polynomials in Theorem \ref{InverseCPAveragingTheorem},
\item averaged characteristic polynomials in Theorem \ref{CPAveragingTheorem},
\item averaged ratio of characteristic polynomials in Theorem \ref{RatioCPAveragingTheorem},
\item the kernel in Theorem \ref{KernelAveragingTheorem},
\end{itemize}
starting from the formulae obtained in Desrosiers and Forrester (Proposition 2 of Ref. \cite{DesrosiersForrester}), and in Forrester and Liu (Proposition 2.5 of Ref. \cite{ForresterLiu}) (for the sake of completeness an alternative method is sketched in the Appendix B). Then, motivated by the applications to the problem considered in \cite{FyodorovEigenvectors}, see Section \ref{connection} for more detail, we perform asymptotic analysis of the aforementioned objects in the joint bulk-edge scaling regime for the special degenerate case where the source matrix $\Omega$ is proportional to the unit matrix, see Section \ref{AsymptoticsSection}. The main results of that Section are summarized by Propositions \ref{ICPAsymptotics}, \ref{CPAsymptotics} and \ref{RATIOAsymptotics}. Lastly, the appendices contain proofs of auxiliary Propositions.

\section{The joint density for singular values}\label{SectionDensity}
Let $X$ be a matrix taken from the ensemble defined by probability measure \eqref{distribution}. Let $\omega_1$, $\ldots$, $\omega_N$ denote the squared singular values of $\Omega$, and by $x_1,\ldots, x_N$ the squared singular values of $X$.
\begin{prop}Assume that the parameters $\omega_1$,$\ldots$, $\omega_N$ are pairwise distinct.
The joint probability density of $(x_1,\ldots,x_N)$ can be written as
\begin{equation}
\label{jpdf}
 P^{(L)}(x_1,\ldots,x_N)=\frac{1}{\mathcal{N}_L} \cdot\triangle\left(x\right)\det\left({}_0 F_1 (1;\omega_kx_j) \right)_{k,j=1}^N
 \prod\limits_{i=1}^Nx_i^L e^{-x_i},
\end{equation}
where $\mathcal{N}_L$ is a normalizing constant and $\triangle\left(x\right)$ is the Vandermonde determinant.
\end{prop}
\begin{proof}
The singular value decomposition for $X$ is
$$
X=U X_D V,
$$
where $U$ and $V$ are unitary matrices of size $N\times N$ and $X_D = \diag (\sqrt{x_1}, \sqrt{x_2}, \cdots , \sqrt{x_N})$. We know that $dX\sim\triangle\left(x\right)^2dUdV$.
In order to obtain the formula in the statement of the Proposition we use the Berezin-Karpelevich integral formula introduced in Refs.\cite{BerezinKarpelevich,Wettig}
\begin{equation}
\begin{split}
&\int\limits_{U(N)}d\mu(U)\int\limits_{U(N)}d\mu(V)e^{\Tr\left(\Omega U X_D V\right)+\Tr\left(\Omega^* V^* X_D U^*\right)} = \frac{\const}{\triangle\left(x\right)}
\frac{\det\left[I_0\left(2\sqrt{\omega_kx_j}\right)\right]_{k,j=1}^N}{\triangle\left(\omega\right)},
\end{split}
\end{equation}
and the well-known relation to the Bessel function of imaginary argument: \begin{equation}{}_0 F_1 (1;xy) = I_0\left(2\sqrt{xy}\right).\end{equation}
\end{proof}
\section{Averages of inverse characteristic polynomials}
\label{ICPSection}
From Section \ref{SectionDensity} we conclude that the joint probability density $P^{(L)}\left(x_1,\ldots,x_N\right)$ can be written as
\begin{equation}\label{Density1}
P^{(L)}\left(x_1,\ldots,x_N\right)=\frac{1}{\mathcal{N}_L}\det\left(\eta_i(x_j)\right)_{i,j=1}^N\det\left(\zeta_i(x_j)\right)_{i,j=1}^N,
\end{equation}
where
\begin{equation}\label{zetaeta}
\eta_i(x)=x^{i-1},\;\;\zeta_i(x)={}_0 F_1 \left (1;\omega_i x \right ) x^{L} e^{-x},
\end{equation}
and the normalization reads
\begin{equation}
\label{NormL}
\mathcal{N}_L = N! \det G,
\end{equation}
with elements of the matrix $G=\left(g_{i,j}\right)_{i,j=1}^N$ given by the integral $g_{i,j}=\int\limits_0^{\infty}\eta_i(x)\zeta_j(x)dx$. We will assume that the parameters $\omega_1$, $\ldots$, $\omega_N$
are pairwise distinct. The following Proposition was proved in Ref. \cite{ForresterLiu} (see Proposition 2.5 of Ref. \cite{ForresterLiu}) for the averages of inverse characteristic polynomials.
\begin{prop}\label{PropositionGeneralFormula} Consider the ensemble defined by equation (\ref{Density1}), and assume that $\eta_i(x)=x^{i-1}$.
Set $G=\left(g_{i,j}\right)_{i,j=1}^N$,
where the matrix entries are defined by
\begin{equation}
\label{gdef}
g_{i,j}=\int\limits_0^{\infty}\eta_i(x)\zeta_j(x)dx,\;\;\; 1\leq i,j\leq N.
\end{equation}
Let $C$ be the inverse of $G$, and let $c_{i,j}$ be the matrix entries of $C^T$, i.e. $C^T=\left(c_{i,j}\right)_{i,j=1}^N$.
Then the following formula holds true
\begin{equation}
\E\left[\prod\limits_{i=1}^N\frac{1}{y-x_i}\right]=\int\limits_{0}^{\infty}du\frac{1}{y-u}\sum\limits_{j=1}^Nc_{N,j}\zeta_j(u).
\end{equation}
where $\E(...)$ denotes averaging over jpdf given by Eq. \eqref{Density1}.
\end{prop}

\subsection{Characterization of the matrix entries of $C$}
We will use the same method as in Ref. \cite{ForresterLiu}  to compute the averages of inverse characteristic polynomials
explicitly.
First, we wish to  obtain a formula characterizing the entries of $C$. For this purpose let us compute the entries of the matrix $G$ explicitly.
We have
\begin{equation}
g_{i,j}=\int\limits_0^{\infty}x^{i+L-1}e^{-x}{}_0 F_1 (1;\omega_j x) dx,\;\;\; i\leq i,j\leq N.
\end{equation}
In order to compute this integral we will use the well-known formula, see e.g. Ref. \cite{GR},
\begin{equation}
\label{Integral}
\int\limits_{0}^{\infty}x^{n+\frac{\nu}{2}}e^{-\alpha x}I_{\nu}\left(2\beta\sqrt{x}\right)dx
=n!\beta^{\nu}e^{\beta^2/\alpha}\alpha^{-n-\nu-1}L_n^{\nu}\left(-\frac{\beta^2}{\alpha}\right),
\end{equation}
where $L_n^{\nu}(x)$ are generalized Laguerre polynomials, and find
\begin{equation}
\label{gij}
g_{i,j}=\left(i+L-1\right)!e^{\omega_j}L_{i+L-1}\left(-\omega_j\right),\;\;\; 1\leq i,j\leq N.
\end{equation}
Taking into account the definition of the matrix entries $c_{i,j}$ (recall that
$c_{i,j}$ are the matrix entries of $C^T$, where $C=G^{-1}$, so $\sum_{j=1}^Nc_{j,i}g_{j,k}=\delta_{i,k}$), we obtain
\begin{equation}\label{characterization1}
\sum\limits_{j=1}^N\left(j+L-1\right)!e^{\omega_k}L_{j+L-1}\left(-\omega_k\right)c_{j,i}=\delta_{i,k},
\end{equation}
where $1\leq i,k\leq N$, and $L\in\left\{0,1,\ldots\right\}$.
\subsection{The case $L=0$}
When $L=0$, the following equation between two polynomials in $u$ holds true
\begin{equation}\label{characterization2}
\sum\limits_{j=1}^N\left(j-1\right)!L_{j-1}\left(u\right)c_{j,i}=e^{-\omega_i}\underset{\tau\neq i}{\prod\limits_{\tau=1}^N}
\frac{-u-\omega_{\tau}}{\omega_i-\omega_{\tau}}
\end{equation}
Indeed, the expressions in the left hand side and in the righthand side of equation (\ref{characterization2}) are both polynomials in $u$ of
degree $N-1$. In addition, we note that
$$
\underset{\tau\neq i}{\prod\limits_{\tau=1}^N}
\frac{\omega_k-\omega_{\tau}}{\omega_i-\omega_{\tau}}=\delta_{i,k}.
$$
Therefore, equation (\ref{characterization1}) implies that equation (\ref{characterization2}) is satisfied at
points $\omega_1$, $\ldots$, $\omega_N$.

We know that the leading term of the Laguerre polynomial is
$$
n!L_{n}^{\nu}(x)=\left(-x\right)^n+\ldots
$$
Therefore, taking $u\rightarrow\infty$ in equation (\ref{characterization2}) we should have
\begin{equation}
c_{N,i}=\frac{e^{-\omega_i}}{\underset{\tau\neq i}{\prod\limits_{\tau=1}^N}\left(\omega_i-\omega_{\tau}\right)},\;\;\; 1\leq i\leq N.
\end{equation}
\begin{rem} Taking into account the explicit expression for the matrix entries $g_{i,j}$  and formula above we see that
the condition $\sum_{k=1}^Nc_{N,k}g_{i,k}=\delta_{N,i}$ is equivalent to
\begin{equation}
\sum\limits_{k=1}^N\frac{1}{\underset{\tau\neq k}{\prod\limits_{\tau=1}^N}\left(\omega_k-\omega_{\tau}\right)}
(i-1)!L_{i-1}\left(-\omega_k\right)=\delta_{N,i},\;\;\; 1\leq i\leq N.
\end{equation}
The equation just written above follows from the formulae for Lagrange interpolation with the monic Laguerre polynomials.
\end{rem}
Now we are ready to apply Proposition \ref{PropositionGeneralFormula}. We have
\begin{equation}
\begin{split}
&\E\left[\prod\limits_{i=1}^N\frac{1}{y-x_i}\right]_{L=0} = \int\limits_{0}^{\infty}du\frac{1}{y-u}\sum\limits_{j=1}^N
\frac{e^{-\omega_j}}{\underset{\tau\neq j}{\prod\limits_{\tau=1}^N}\left(\omega_j-\omega_{\tau}\right)}
{}_0 F_1 (1;\omega_j u) e^{-u}
\end{split}
\end{equation}
or
\begin{equation}\label{FinalFormulaL0}
\begin{split}
\E\left[\prod\limits_{i=1}^N\frac{1}{y-x_i}\right]_{L=0}
=\frac{1}{2\pi i}
\int\limits_{0}^{\infty}du\frac{e^{-u}}{y-u}
\int\limits_{C}\frac{e^{-v}{}_0 F_1 \left (1;vu \right )dv}{\prod_{k=1}^N\left(v-\omega_k\right)},
\end{split}
\end{equation}
where the counterclockwise contour $C$ encircles $\omega_1$, $\ldots$, $\omega_N$.
\begin{rem}
\label{FyodorovRemark}
Assume that $\Omega=z \mathbf{1}_N$, $z\in\C$. Then $\Omega^*\Omega=|z|^2 \mathbf{1}_N$, and $\omega_1=\ldots=\omega_N=|z|^2$.
In this degenerate case (denoted with a subscript "$\text{deg}$") we obtain the formula
\begin{equation}\label{FinalFormulaL0FyodorovCase}
\begin{split}
&\E\left[\prod\limits_{i=1}^N\frac{1}{y-x_i}\right]_{\substack{\text{deg}\\L=0}}=\frac{1}{2\pi i}
\int\limits_{0}^{\infty}du\frac{e^{-u}}{y-u}
\int\limits_{C}\frac{e^{-v}{}_0 F_1 \left (1;vu \right ) dv}{\left(v-|z|^2\right)^N}.
\end{split}
\end{equation}
It is not hard to check by direct calculations that this formula holds true at $N=1$.
\end{rem}
We conclude that equation (\ref{FinalFormulaL0}) solves the problem of finding the mean inverse characteristic polynomial for the case $L=0$ for a general source matrix,
and equation (\ref{FinalFormulaL0FyodorovCase}) solves the corresponding degenerate problem considered in \cite{FyodorovEigenvectors}.

\subsection{The case of a general $L$}
In this Section we prove Theorem \ref{InverseCPAveragingTheorem} which gives an explicit expression for the averages of inverse characteristic polynomial in the case of a general $L$.
\begin{thm}\label{InverseCPAveragingTheorem} Consider the probability distribution on the space of complex matrices of size $N\times N$ defined by equation (\ref{distribution}).\\
\textbf{(A)} Let $\omega_1$, $\ldots$, $\omega_N$ be the squared singular values of $\Omega$, and assume that $L\geq 1$. We have
\begin{equation}
\label{InverseCPAveragingForm}
\begin{split}
&\E\left[\prod\limits_{i=1}^N\frac{1}{y-x_i}\right] = \frac{1}{\tilde{\mathcal{N}}_L} \frac{1}{2\pi i} \int\limits_{0}^{\infty}du\frac{e^{-u}u^L}{y-u}
\int\limits_{C}\frac{e^{-v}{}_0 F_1 \left (1;vu \right )dv}{\prod\limits_{k=1}^N\left(v-\omega_k\right)}
\int\limits_0^{\infty}...\int\limits_0^{\infty}\frac{\prod\limits_{i=1}^L \prod\limits_{j=1}^N\left(t_i+\omega_j\right)}{\prod\limits_{i=1}^L(t_i+v)}
\triangle^2\left(t\right)\prod_{i=1}^L e^{-t_i}dt_i,
\end{split}
\end{equation}
where the counter-clockwise contour $C$ encircles $\omega_1$, $\ldots$, $\omega_N$. The normalization constant is equal to
\begin{equation}
\label{tildeNormL}
\begin{split}
\tilde{\mathcal{N}}_L = {\int\limits_0^{\infty}...\int\limits_0^{\infty}\prod_{i=1}^L\prod_{j=1}^N\left(t_i+\omega_j\right)
\triangle^2(t) \prod_{i=1}^L e^{-t_i}dt_i}
\end{split}
\end{equation}
and is related to the normalization $\mathcal{N}_L$ of the density \eqref{jpdf} as
\begin{align}
\label{NormLRelation}
\mathcal{N}_L = \frac{(-1)^{N(N-1)} N! \triangle(\omega) \prod\limits_{i=1}^N e^{\omega_i} }{L!\prod\limits_{j=1}^{L-1} (j)!^2} \tilde{\mathcal{N}}_L.
\end{align}
\textbf{(B)} Assume that $\Omega=z \mathbf{1}_N$, $z\in\C$. Then $\Omega^*\Omega=|z|^2 \mathbf{1}_N$, and $\omega_1=\ldots=\omega_N=|z|^2$.
In this case we obtain the formula
\begin{equation}
\label{ICPDegForm}
\begin{split}
&\E\left[\prod\limits_{i=1}^N\frac{1}{y-x_i}\right]_{\text{deg}} = \frac{1}{\tilde{\mathcal{N}}_L} \frac{1}{2\pi i} \int\limits_{0}^{\infty}du\frac{e^{-u}u^L}{y-u}
\int\limits_{C}\frac{e^{-v}{}_0 F_1 \left (1;vu \right )dv}{\left(v-|z|^2\right)^N}
\int\limits_0^{\infty}...\int\limits_0^{\infty}\frac{\prod\limits_{i=1}^L \left(t_i+|z|^2\right)^N}{\prod\limits_{i=1}^L(t_i+v)}
\triangle^2\left(t\right)\prod_{i=1}^L e^{-t_i}dt_i,
\end{split}
\end{equation}
the counter-clockwise contour $C$ encircles $|z|^2$ and the corresponding normalization is $\tilde{\mathcal{N}}_L = \int_0^{\infty}...\int_0^{\infty}\prod_{i=1}^L \left(t_i+|z|^2\right)^N
\triangle^2(t) \prod_{i=1}^L e^{-t_i}dt_i$.
\end{thm}
\begin{proof}It is convenient to assume that   the parameters $\omega_1$, $\ldots$, $\omega_N$ are pairwise distinct.
We begin from the observation that Proposition \ref{PropositionGeneralFormula} can be stated in an equivalent form.
\begin{prop}
With the same notation as in the statement of Proposition \ref{PropositionGeneralFormula} the formula for
the average of an inverse characteristic polynomial can be written as ratio of determinants, namely
\begin{equation}\label{GeneralFormula1}
\E\left[\prod\limits_{i=1}^N\frac{1}{y-x_i}\right]=\frac{1}{\det G}
\left|\begin{array}{cccc}
        g_{1,1} & g_{1,2} & \ldots & g_{1,N} \\
        \vdots & \vdots & \ddots & \vdots \\
        g_{N-1,1} & g_{N-1,2} & \ldots & g_{N-1,N} \\
        \int\limits_0^{\infty}du\frac{\zeta_1(u)}{y-u} & \int\limits_0^{\infty}du\frac{\zeta_2(u)}{y-u} & \ldots &  \int\limits_0^{\infty}du\frac{\zeta_N(u)}{y-u}
      \end{array}
\right|.
\end{equation}
\end{prop}
\begin{proof}
See Desrosiers and Forrester \cite{DesrosiersForrester} (Proposition 2), Forrester and Liu \cite{ForresterLiu} (Proposition 2.5).
\end{proof}
By Laplace expansion, equation (\ref{GeneralFormula1}) can be rewritten as
\begin{equation}\label{GeneralFormula2}
\begin{split}
&\E\left[\prod\limits_{i=1}^N\frac{1}{y-x_i}\right]\\
&=\int\limits_0^{\infty}du\frac{1}{y-u}
\sum\limits_{k=1}^N\left(-1\right)^{N-k}\zeta_k(u)
\frac{\left|\begin{array}{cccccc}
              g_{1,1} & \ldots & g_{1,k-1} & g_{1,k+1} & \ldots & g_{1,N} \\
              \vdots & \ddots & \vdots & \vdots & \ddots & \vdots \\
              g_{N-1,1} & \ldots & g_{N-1,k-1} & g_{N-1,k+1} & \ldots & g_{N-1,N}
            \end{array}
\right|}{\left|\begin{array}{ccc}
           g_{1,1} & \ldots & g_{1,N} \\
           \vdots & \ddots & \vdots \\
           g_{N,1} & \ldots & g_{N,N}
         \end{array}
\right|}.
\end{split}
\end{equation}
In our case the matrix entries $g_{i,j}$ given by equation \eqref{gij} can be written in terms of the Laguerre
monic polynomials $\left\{\pi_k\left(x\right)\right\}_{k=0}^{\infty}$ as
\begin{equation}\label{gijasMonicLaguerrePolynomial}
g_{i,j}=\left(-1\right)^{i+L-1}\pi_{i+L-1}\left(-\omega_j\right)e^{\omega_j},\;\; 1\leq i,j\leq N,
\end{equation}
where the Laguerre monic polynomials $\pi_k(x)=x^k+\ldots $ are defined in terms of the Laguerre polynomials as
$\pi_k(x)=k!(-1)^kL_k(x)$. So in our case formula (\ref{GeneralFormula2}) takes the form
\begin{equation}\label{GeneralFormula3}
\begin{split}
&\E\left[\prod\limits_{i=1}^N\frac{1}{y-x_i}\right]=(-1)^{L-1}\int\limits_0^{\infty}\frac{du}{y-u}u^L e^{-u}
\sum\limits_{k=1}^N\left(-1\right)^{k}{}_0 F_1 \left (1;\omega_k u \right )e^{-\omega_k} \\
&\times\frac{\left|\begin{array}{cccccc}
              \pi_L\left(-\omega_1\right) & \ldots & \pi_L\left(-\omega_{k-1}\right) & \pi_L\left(-\omega_{k+1}\right) & \ldots & \pi_L\left(-\omega_N\right) \\
              \vdots & \ddots & \vdots & \vdots & \ddots & \vdots \\
              \pi_{N+L-2}\left(-\omega_1\right) & \ldots & \pi_{N+L-2}\left(-\omega_{k-1}\right) & \pi_{N+L-2}\left(-\omega_{k+1}\right) & \ldots & \pi_{N+L-2}\left(-\omega_N\right)
            \end{array}
\right|}{\left|\begin{array}{ccc}
          \pi_L\left(-\omega_1\right) & \ldots & \pi_L\left(-\omega_1\right) \\
           \vdots & \ddots & \vdots \\
          \pi_{N+L-1}\left(-\omega_1\right) & \ldots & \pi_{N+L-1}\left(-\omega_N\right)
         \end{array}
\right|}.
\end{split}
\end{equation}
Note the the ratio of determinants in formula (\ref{GeneralFormula3}) can be rewritten as
\begin{equation}\label{RatioDeterminants}
\frac{\left|\begin{array}{ccc}
                        \pi_L\left(-\omega_1\right) & \ldots & \pi_{L+N-2}\left(-\omega_1\right) \\
                        \vdots & \ddots & \vdots \\
                        \pi_L\left(-\omega_{k-1}\right) & \ldots & \pi_{L+N-2}\left(-\omega_{k-1}\right) \\
                        \pi_L\left(-\omega_{k+1}\right) & \ldots & \pi_{L+N-2}\left(-\omega_{k+1}\right) \\
                        \vdots & \ddots & \vdots \\
                        \pi_L\left(-\omega_N\right) & \ldots & \pi_{L+N-2}\left(-\omega_N\right)
                      \end{array}
\right|}{\left|\begin{array}{ccc}
          \pi_L\left(-\omega_1\right) & \ldots & \pi_{N+L-1}\left(-\omega_1\right) \\
           \vdots & \ddots & \vdots \\
          \pi_{L}\left(-\omega_N\right) & \ldots & \pi_{N+L-1}\left(-\omega_N\right)
         \end{array}
\right|}.
\end{equation}
In order to evaluate the ratio of determinants just written above  we will use the following result
by Brezin and Hikami \cite{BrezinHikami}.
\begin{prop}\label{PropositionBrezinHikami}Let $d\alpha$ be a measure with finite moments $\int_{\R}|t|^kd\alpha(t)<\infty$, $k=0,1,2,\ldots$.
Let $\pi_j(t)=t^j+\ldots $ denote the $j$th monic orthogonal polynomial with respect to the measure
$d\alpha$
$$
\int\limits_{\R}\pi_j(t)\pi_k(t)d\alpha(t)=c_jc_k\delta_{j,k},\;\; j,k>0.
$$
Then
\begin{equation}
\begin{split}
&\left|\begin{array}{ccc}
        \pi_n\left(\mu_1\right) & \ldots & \pi_{n+m-1}\left(\mu_1\right) \\
        \vdots & \ddots & \vdots  \\
        \pi_n\left(\mu_m\right) & \ldots & \pi_{n+m-1}\left(\mu_m\right)
      \end{array}
\right| = \triangle\left(\mu_1,...,\mu_m\right)\frac{\int\limits_{\R}...\int\limits_{\R}\prod\limits_{i=1}^n\prod\limits_{j=1}^m\left(\mu_j-t_i\right)\triangle^2\left(t_1,...,t_n\right)
\prod\limits_{i=1}^n d\alpha(t_i)}{\int\limits_{\R}...\int\limits_{\R}\triangle^2\left(t_1,...,t_n\right)
\prod\limits_{i=1}^n d\alpha(t_i)},
\end{split}
\end{equation}
where
$$
\triangle\left(\mu_1,...,\mu_m\right)=\prod\limits_{1\leq i<j\leq m}\left(\mu_j-\mu_i\right).
$$
\end{prop}
\begin{proof}
See  Brezin and Hikami \cite{BrezinHikami}.
\end{proof}
We apply the formula stated in Proposition \ref{PropositionBrezinHikami} to both numerator (with $n=L$, $m=N-1$,
$\mu_1=-\omega_1$,..., $\mu_{k-1}=-\omega_{k-1}$, $\mu_{k}=-\omega_{k+1},..., \mu_{N-1}=-\omega_N$) and denominator (with $n=L$, $m=N$,
$\mu_1=-\omega_1$,..., $\mu_{N}=-\omega_{N}$) of expression \eqref{RatioDeterminants}:
\begin{equation*}
\begin{split}
 & \frac{\triangle\left(-\omega_1,...,-\omega_{k-1},-\omega_{k+1},...,-\omega_N\right)}{\triangle\left(-\omega_1,...,-\omega_N\right)} \frac{\int\limits_{0}^{\infty}...\int\limits_{0}^{\infty}\prod\limits_{i=1}^L\prod\limits_{\substack{j=1\\j\neq k}}^N
\left(-\omega_j-t_i\right)\triangle^2\left(t_1,...,t_L\right)
\prod\limits_{i=1}^L e^{-t_i} dt_i } {\int\limits_{0}^{\infty}...\int\limits_{0}^{\infty}\prod\limits_{i=1}^L\prod\limits_{j=1}^N
\left(-\omega_j-t_i\right)\triangle^2\left(t_1,...,t_L\right)
\prod\limits_{i=1}^L e^{-t_i} dt_i } = \\
& = (-1)^{1-N-L} \frac{\triangle\left(\omega_1,...,\omega_{k-1},\omega_{k+1},...,\omega_N\right)}{\triangle\left(\omega_1,...,\omega_N\right)} \frac{\int\limits_{0}^{\infty}...\int\limits_{0}^{\infty}\prod\limits_{i=1}^L\prod\limits_{\substack{j=1\\j\neq k}}^N
\left(\omega_j+t_i\right)\triangle^2\left(t_1,...,t_L\right)
\prod\limits_{i=1}^L e^{-t_i} dt_i}{\int\limits_{0}^{\infty}...\int\limits_{0}^{\infty}\prod\limits_{i=1}^L\prod\limits_{j=1}^N
\left(\omega_j+t_i\right)\triangle^2\left(t_1,...,t_L\right)
\prod\limits_{i=1}^L e^{-t_i} dt_i},
\end{split}
\end{equation*}
where we used $\triangle\left(-\omega_1,...,-\omega_N\right) = (-1)^{N(N-1)/2} \triangle\left(\omega_1,\ldots,\omega_N\right)$.
Taking into account that
\begin{equation}
\frac{\triangle\left(\omega_1,...,\omega_{k-1},\omega_{k+1},...,\omega_N\right)}{\triangle\left(\omega_1,...,\omega_N\right)}
=\frac{1}{\prod\limits_{i=1}^{k-1}\left(\omega_k-\omega_i\right)\prod\limits_{i=k+1}^{N}\left(\omega_i-\omega_k\right)}
=\frac{(-1)^{N-k}}{\prod\limits_{\substack{j=1\\j \neq k}}^N \left(\omega_k-\omega_j\right)},
\end{equation}
we get a representation of the ratio of two determinants with orthogonal polynomial entries
in terms of multiple integrals
\begin{equation*}\label{RatioIntegrals}
\begin{split}
&\frac{\left|\begin{array}{ccc}
                        \pi_L\left(-\omega_1\right) & \ldots & \pi_{L+N-2}\left(-\omega_1\right) \\
                        \vdots & \ddots & \vdots \\
                        \pi_L\left(-\omega_{k-1}\right) & \ldots & \pi_{L+N-2}\left(-\omega_{k-1}\right) \\
                        \pi_L\left(-\omega_{k+1}\right) & \ldots & \pi_{L+N-2}\left(-\omega_{k+1}\right) \\
                        \vdots & \ddots & \vdots \\
                        \pi_L\left(-\omega_N\right) & \ldots & \pi_{L+N-2}\left(-\omega_N\right)
                      \end{array}
\right|}{\left|\begin{array}{ccc}
          \pi_L\left(-\omega_1\right) & \ldots & \pi_{N+L-1}\left(-\omega_1\right) \\
           \vdots & \ddots & \vdots \\
          \pi_{L}\left(-\omega_N\right) & \ldots & \pi_{N+L-1}\left(-\omega_N\right)
         \end{array}
\right|} =
 \frac{(-1)^{k+L-1}}{\prod\limits_{\substack{j=1\\ j\neq k}}^N \left(\omega_k-\omega_j\right)} \frac{\int\limits_{0}^{\infty}...\int\limits_{0}^{\infty}\prod\limits_{i=1}^L\prod\limits_{\substack{j=1\\j\neq k}}^N
\left(\omega_j+t_i\right)\triangle^2\left(t_1,...,t_L\right)
\prod\limits_{i=1}^L e^{-t_i} dt_i}{\int\limits_{0}^{\infty}...\int\limits_{0}^{\infty}\prod\limits_{i=1}^L\prod\limits_{j=1}^N
\left(\omega_j+t_i\right)\triangle^2\left(t_1,...,t_L\right)
\prod\limits_{i=1}^L e^{-t_i} dt_i}.
\end{split}
\end{equation*}
In order to get the formula in the statement of Theorem \ref{InverseCPAveragingTheorem}, (A) insert the righthand side of equation
(\ref{RatioIntegrals}) into the expression in the righthand side of formula (\ref{GeneralFormula3}), and use the Basic Residue Theorem.
By continuity, we can remove  the condition that all parameters  $\omega_1$, $\ldots$, $\omega_N$ are pairwise distinct, and, in particular,
to obtain the formula in the statement of Theorem \ref{InverseCPAveragingTheorem}, (B).

The normalization constant $\tilde{\mathcal{N}}_L$ is read off from the denominator of the above formula
\begin{align}
\tilde{\mathcal{N}}_L = \int\limits_{0}^{\infty}...\int\limits_{0}^{\infty}\prod\limits_{i=1}^L\prod\limits_{j=1}^N
\left(\omega_j+t_i\right)\triangle^2\left(t_1,...,t_L\right)
\prod\limits_{i=1}^L e^{-t_i} dt_i,
\end{align}
and the relation to $\mathcal{N}_L$ given by expression \eqref{jpdf} is found from the definition \eqref{NormL} and the formula \eqref{gijasMonicLaguerrePolynomial}
\begin{equation*}
\mathcal{N}_L = N! (-1)^{\frac{N(N+1)}{2}  + N(L-1)} \prod_{i=1}^N e^{\omega_i} \det \left ( \pi_{i+L-1}(-\omega_j) \right )_{i,j=1}^N.
\end{equation*}
Similarly as before, the determinant is found by applying Proposition \ref{PropositionBrezinHikami}
\begin{equation*}
\begin{split}
&\det \left ( \pi_{i+L-1}(-\omega_j) \right )_{i,j=1}^N = \frac{\triangle\left(-\omega_1,...,-\omega_N\right) \int\limits_{0}^{\infty}...\int\limits_{0}^{\infty}\prod\limits_{i=1}^L\prod\limits_{j=1}^N
\left(-\omega_j-t_i\right)\triangle^2\left(t_1,...,t_L\right) \prod\limits_{i=1}^L e^{-t_i} dt_i}{\int\limits_{0}^{\infty}...\int\limits_{0}^{\infty}\triangle^2\left(t_1,...,t_L\right)
\prod\limits_{i=1}^L e^{-t_i} dt_i},
\end{split}
\end{equation*}
where the numerator is proportional to $\tilde{\mathcal{N}}_L$ and denominator is a Mehta-type integral formula $\int\limits_{0}^{\infty}...\int\limits_{0}^{\infty}\triangle^2\left(t_1,...,t_L\right)
\prod\limits_{i=1}^L e^{-t_i} dt_i = L!\prod\limits_{j=1}^{L-1} j!^2$. Combining above formulas gives the relation \eqref{NormLRelation}.
\end{proof}
For completeness, in Appendix \ref{ICPalternative} we provide an alternative proof of the Proposition \ref{InverseCPAveragingTheorem}.
\subsection{Connection with the statistics of eigenvectors of a complex Ginibre matrix} \label{connection}
One of the motivations in studying the deformed ensemble given by equation \eqref{distribution} is its connection to the statistics of eigenvectors in complex Ginibre ensemble as elucidated in Ref. \cite{FyodorovEigenvectors}. It turns out that the joint eigenvalue-eigenvector probability density function:
\begin{align}
\mathcal{P}^{(2)} (t,z) = \E \left [ \sum_{i=1}^N \delta (O_{ii} - 1 - t) \delta^{(2)} (z-\lambda_i) \right ]_{\text{Ginibre}},
\end{align}
where $\lambda_i$ are the complex eigenvalues and the self-overlap matrix $O_{ii} = \left < L_i | L_i \right > \left < R_i | R_i \right > $ consists of left $\left |{L_i} \right >$ and right $\left | {R_i} \right >$ eigenvectors. This correlation function is expressible as a Laplace transform
\begin{align}
\mathcal{P}^{(2)} (t,z) \sim \int_0^\infty dp e^{-pt} \mathcal{D}^{(L=2)}_{N-1,\beta=2}(z,p),
\end{align}
where the integrand is a ratio of determinants:
\begin{equation}
\mathcal{D}^{(L)}_{N,\beta=2}(z,p) = \frac{1}{c}\int dX e^{-\Tr X^\dagger X} \frac{\det \left [(\bar{z} \textbf{1}_N - X^\dagger)(z\textbf{1}_N - X)\right ]^L}{\det \left [ p \textbf{1}_N + (\bar{z} \textbf{1}_N - X^\dagger)(z\textbf{1}_N - X) \right ]},
\end{equation}
with normalization $c = \int dX e^{-\Tr X^\dagger X}$ and the real parameter $p>0$. By setting $X \to X + z\textbf{1}_N$, this correlation function reduces to an inverse characteristic polynomial averaged over the jpdf introduced in equation \eqref{distribution} for a degenerate source term $\Omega$ corresponding to $\omega_1=...=\omega_N = |z|^2$:
\begin{align}
\label{DInverseCPRelation}
 \mathcal{D}^{(L)}_{N,\beta=2}(z,p) = \frac{\mathcal{N}_L}{\mathcal{N}_0} \int dx P^{(L)}(x) \prod_{i=1}^N \frac{1}{p + x_i} = \frac{\mathcal{N}_L}{\mathcal{N}_0} \E \left [ \prod_{i=1}^N \frac{1}{p + x_i} \right ]_{\text{deg}},
 \end{align}
 where $P^{(L)}(X)$ is the jpdf given by formula \eqref{jpdf} and $\mathcal{N}_0$ is the special case of normalization given in \eqref{NormL}. The renormalizing factor is taking into account the switch of conventions from viewing $\mathcal{D}^{(L)}$ as the ratio of characteristic polynomials $\E \left [ \frac{\det^L(X^\dagger X)}{\det (p \mathbf{1}_N + X^\dagger X)} \right ]_{\text{Ginibre}}$ averaged over the (shifted by $z$) complex Ginibre ensemble to the inverse characteristic polynomial $\E \left [ \frac{1}{\det (p\mathbf{1}_N + X^\dagger X)} \right ]_{P^{(L)}}$ averaged over the deformed ensemble $P^{(L)}$ defined by formula \eqref{distribution}.

We also comment on how the discussed connection to the complex Ginibre ensemble can mix the nomenclature one chooses in the asymptotic analysis, a task addressed in Section \ref{AsymptoticsSection}. The problem arises since now the standard bulk and edge regimes have an ambiguous interpretation. The averages are evaluated with parameters probing either inside (bulk scaling) or the edge (edge scaling) of the underlying spectral density of 1) the matrix $X$ (Ginibre interpretation) or 2) the matrix $X^\dagger X$ (Wishart-Laguerre/chiral GUE intepretation). The choice is not evident in this case as the matrix model and in Sec. \ref{AsymptoticsSection} we consider a \emph{joint bulk-edge regime} where we inspect the bulk scaling of $|z|^2 \sim N$ along with the edge regime of the characteristic polynomials arguments scaling as $p \sim N^{-1}$.

Explicit expressions for the object defined in the Eq.\eqref{DInverseCPRelation} in simplest cases $L=0,1,2$ were found in Ref. \cite{FyodorovEigenvectors} in the framework
of a supersymmetry approach. In particular,
\begin{align}
\mathcal{D}^{(0)}_{N,2}(z,p) & = \frac{1}{(N-1)!} \int_0^\infty \frac{e^{-pt - \frac{t|z|^2}{1+t}}}{t+1} \left ( \frac{t}{1+t} \right )^{N-1}, \label{DL0}\\
\mathcal{D}^{(1)}_{N,2}(z,p) & = \frac{1}{(N-1)!} e^{|z|^2} \int_0^\infty \frac{e^{-pt}}{t(1+t)} e^{-\frac{t|z|^2}{1+t}} \left ( \frac{t}{1+t} \right )^{N} \left ( \Gamma(N+1,|z|^2) - \frac{t}{1+t} |z|^2 \Gamma(N,|z|^2) \right ) .\label{DL1}
\end{align}

To see how these expressions follow in the present approach we use the following
\begin{rem}
\label{DRelation}
For $L\geq 0$ the object defined in Ref. \cite{FyodorovEigenvectors} and related to the averaged inverse characteristic polynomial by formula \eqref{DInverseCPRelation} is equal to
\begin{equation}
\begin{split}
 \mathcal{D}^{(L)}_{N,\beta=2}(z,p) = & \frac{(-1)^{N-1}}{L! \prod\limits_{j=1}^{L-1} j!^2} \frac{1}{2\pi i} \int\limits_{0}^{\infty}du\frac{e^{-u}u^L}{p+u}
\int\limits_{C}\frac{e^{-v}{}_0 F_1 \left (1;vu \right )dv}{\left(v-|z|^2\right)^N} \\
& \times \int\limits_0^{\infty}...\int\limits_0^{\infty}\frac{\prod\limits_{i=1}^L \left(t_i+|z|^2\right)^N}{\prod\limits_{i=1}^L(t_i+v)}
\triangle\left(t\right)^2\prod_{i=1}^L e^{-t_i}dt_i.
\end{split}
\end{equation}
\end{rem}
\begin{proof}
This follows directly from the relation \eqref{DInverseCPRelation} and the correction factor is found using the formula \eqref{NormLRelation}
\begin{align}
\frac{\mathcal{N}_L}{\mathcal{N}_0} = \frac{\tilde{\mathcal{N}}_L}{L! \prod_{j=1}^{L-1} j!^2},
\end{align}
where $\tilde{\mathcal{N}}_0 = 1$.
\end{proof}

In the $L=0$ case the formula \eqref{DL0} is related to Remark \ref{DRelation} by the following Proposition.
\begin{prop}
\label{L0Equivalence}
For $p>0$ the following formula holds true
\begin{equation}
\frac{1}{2\pi i} \int\limits_{0}^{\infty}du\frac{e^{-u}}{p+u}
\int\limits_{C}\frac{e^{-v}{}_0 F_1 \left (1;vu \right )dv}{\left(v-|z|^2\right)^N} = \frac{(-1)^{N-1}}{(N-1)!}\int\limits_0^{\infty} \frac{dt}{1+t} e^{-tp}\, \left(\frac{t}{1+t}\right)^{N-1} e^{-|z|^2\frac{t}{1+t}}.
\end{equation}
\end{prop}
\begin{proof}
We use $ \frac{1}{p+u}=\int\limits_0^{\infty}dt e^{-t(p+u)}$ to arrive at
\[
\text{l.h.s.} = \frac{1}{2\pi i}\int_0^{\infty} dt e^{-tp} \int\limits_C \frac{dv}{(v-|z|^2)^N} e^{-v} \int_0^{\infty} du  e^{-u(1+t)} {}_0 F_1 \left (1;vu \right ).
\]
To evaluate the last integral we use again the identity \eqref{Integral} for $\nu=0,\, n=0,\, \beta=\sqrt{v}, \, \alpha=1+t$ and the integral in question is equal to $\frac{1}{1+t} e^{\frac{v}{1+t}}$. Lastly, the contour integral is evaluated:
\begin{align*}
\text{l.h.s.} & = \int_0^{\infty} \frac{dt}{1+t} e^{-tp}\, \frac{1}{2\pi i}\int\limits_C \frac{dv}{(v-|z|^2)^N} e^{-v\frac{t}{1+t}} = \frac{(-1)^{N-1}}{(N-1)!}\int\limits_0^{\infty} \frac{dt}{1+t} e^{-tp}\, \left(\frac{t}{1+t}\right)^{N-1} e^{-|z|^2\frac{t}{1+t}}.
\end{align*}
\end{proof}

Following the same lines and exploiting the identity (\ref{Integral}) one can prove for general integer $L\ge 0$ the following equivalent
representations (cf. (\ref{qGrel}) later on):
\begin{prop}\label{mynorm}
\begin{equation}\label{Gen2mynorm}
 \mathcal{D}^{(L)}_{N,\beta=2}(z,p)=\int_0^{\infty} \frac{dt}{(1+t)^{L+1}} e^{-tp}\, \mathcal{G}^{(L)}_N\left(|z|^2,\frac{t}{1+t}\right)
\end{equation}
where we defined the following function of $\rho=|z|^2$ and $\tau=t/(1+t)$:
\begin{equation}\label{Gen3mynorm}
 \mathcal{G}^{(L)}_N(\rho,\tau)=(-1)^N L!\,\int_0^{\infty}dt_1\ldots \int_0^{\infty}dt_L \Delta^2(t_1,\ldots,t_L) \prod_{k=1}^L (t_k+\rho)^N\,e^{-t_k}
\end{equation}
\[
\times\frac{1}{2\pi i}\oint_{Re(v)>0} \frac{dv e^{-v\,\tau}}{(v-\rho)^N}L_L\left(-v(1-\tau\right))\frac{1}{\prod_{k=1}^L (t_k+v)}
\]
which also can be presented in an explicitly real form:
\begin{equation}\label{G1amynorm}
\mathcal{G}^{(L)}_N(\rho,\tau)=\frac{(-1)^N L!}{(N-1)!}\int_0^{\infty}dt_1\ldots \int_0^{\infty}dt_L \Delta^2(t_1,\ldots,t_L) \prod_{k=1}^L (t_k+\rho)^N\,e^{-t_k}
\end{equation}
\[
\times\frac{d^{N-1}}{d^{N-1}\rho}\left[ e^{-\rho\,\tau}\,L_L\left(-\rho(1-\tau)\right)\frac{1}{\prod_{k=1}^L (t_k+\rho)}\right]
\]
\end{prop}
The above can be used for a verification of the $L=1$ case, which is a similar but substantially longer calculation than for $L=0$ and for this sake is  relegated to the Appendix \ref{AppL1}. The 'joint bulk-edge asymptotics' limit of  $\mathcal{G}^{(L)}_N(\rho,\tau)$ relevant for the context of Ref. \cite{FyodorovEigenvectors} is evaluated in the Section 7 of the present paper, see Eq. (\ref{GLafinmynorm}).

\section{Average of characteristic polynomials}
\label{CPSection}
We turn our attention to the average characteristic polynomial evaluated for a general integer parameter $L\geq 1$.
\begin{thm}\label{CPAveragingTheorem} Consider the probability distribution on the space of complex matrices of size $N\times N$ defined by equation (\ref{distribution}).\\
\textbf{(A)} Let $\omega_1$, $\ldots$, $\omega_N$ be the squared singular values of $\Omega$, and assume that $L\geq 1$. We have
\begin{equation}
\begin{split}
\E\left[\prod\limits_{i=1}^N\left(z-x_i\right)\right] = & \frac{(-1)^{N+L}}{\tilde{\mathcal{N}}_L}\frac{e^{z}}{z^L}\int\limits_0^{\infty}dye^{-y}{}_0F_1(1;-zy)\prod_{j=1}^N\left(\omega_j+y\right) \\
&\times \int\limits_{0}^{\infty}...\int\limits_{0}^{\infty}\prod\limits_{i=1}^L\prod\limits_{j=1}^N
\left(\omega_j+t_i\right) \triangle^2\left(t\right)
\prod_{i=1}^L (y-t_i) e^{-t_i}dt_i.
\end{split}
\end{equation}
where the normalization coefficient $\tilde{\mathcal{N}}_L$ is given by equation \eqref{tildeNormL}.\\
\textbf{(B)} Assume that $\Omega=z \mathbf{1}_N$, $z\in\C$. Then $\Omega^*\Omega=|z|^2\mathbf{1}_N$, and $\omega_1=\ldots=\omega_N=|z|^2$.
In this degenerate case we obtain the formula
\begin{equation}
\label{CPDegForm}
\begin{split}
\E\left[\prod\limits_{i=1}^N\left(z-x_i\right)\right]_{\text{deg}} = & \frac{(-1)^{N+L}}{\tilde{\mathcal{N}}_L}\frac{e^{z}}{z^L}\int\limits_0^{\infty}dye^{-y}{}_0F_1(1;-zy) \left(|z|^2+y\right)^N  \\
&\times \int\limits_{0}^{\infty}...\int\limits_{0}^{\infty}\prod\limits_{i=1}^L
\left(|z|^2+t_i\right)^N \triangle^2\left(t\right)
\prod_{i=1}^L (y-t_i) e^{-t_i}dt_i.
\end{split}
\end{equation}
with the normalization factor $\tilde{\mathcal{N}}_L = \int_0^{\infty}...\int_0^{\infty}\prod_{i=1}^L \left(t_i+|z|^2\right)^N \triangle^2(t) \prod_{i=1}^L e^{-t_i}dt_i$.
\end{thm}
\begin{proof} We assume that the parameters $\omega_1$, $\ldots$, $\omega_N$ are pairwise distinct. We begin from the observation that Proposition \ref{PropositionGeneralFormula} can be stated in an equivalent form.

\begin{prop} We have
\begin{equation}
\E\left[\prod\limits_{i=1}^N\left(z-x_i\right)\right] = \frac{1}{\det G}
\left|\begin{array}{cccc}
  g_{1,1} & \ldots & g_{1,N} & \eta_1(z) \\
  g_{2,1} & \ldots & g_{2,N} & \eta_2(z) \\
  \vdots & \ddots & \vdots  & \vdots \\
  g_{N+1,1} & \ldots & g_{N+1,N} & \eta_{N+1}(z)
\end{array}
\right|,
\end{equation}
where $\eta_i(x)=x^{i-1}+\ldots$ is any system of monic polynomials and the matrix $g_{i,j}$ is given by equation \eqref{gdef}.
\end{prop}
\begin{proof}
See Desrosiers and Forrester \cite{DesrosiersForrester}, Proposition 2.
\end{proof}
Recall that in our case the matrix entries $g_{i,j}$ can be rewritten in terms of the monic Laguerre polynomials, see equation (\ref{gijasMonicLaguerrePolynomial}),
so we can write
\begin{equation}
\begin{split}
\E\left[\prod\limits_{i=1}^N\left(z-x_i\right)\right] =
\frac{\left|\begin{array}{cccc}
              (-1)^L\pi_L\left(-\omega_1\right)e^{\omega_1} & \ldots & (-1)^L\pi_L\left(-\omega_N\right)e^{\omega_N} & 1 \\
              (-1)^{L+1}\pi_{L+1}\left(-\omega_1\right)e^{\omega_1} & \ldots & (-1)^{L+1}\pi_{L+1}\left(-\omega_N\right)e^{\omega_N} & z \\
              \vdots & \ddots & \vdots & \vdots \\
              (-1)^{L+N}\pi_{L+N}\left(-\omega_1\right)e^{\omega_1} & \ldots & (-1)^{L+N}\pi_{L+N}\left(-\omega_N\right)e^{\omega_N} & z^N \\
            \end{array}
\right|}{\left|\begin{array}{ccc}
              (-1)^L\pi_L\left(-\omega_1\right)e^{\omega_1} & \ldots & (-1)^L\pi_L\left(-\omega_N\right)e^{\omega_N}  \\
              (-1)^{L+1}\pi_{L+1}\left(-\omega_1\right)e^{\omega_1} & \ldots & (-1)^{L+1}\pi_{L+1}\left(-\omega_N\right)e^{\omega_N}  \\
              \vdots & \ddots &  \vdots  \\
              (-1)^{L+N-1}\pi_{L+N-1}\left(-\omega_1\right)e^{\omega_1} & \ldots & (-1)^{L+N-1}\pi_{L+N-1}\left(-\omega_N\right)e^{\omega_N}  \\
            \end{array}\right|},
\end{split}
\end{equation}
or
\begin{equation}\label{A1}
\begin{split}
\E\left[\prod\limits_{i=1}^N\left(z-x_i\right)\right] = \frac{(-1)^L}{z^L}
\frac{\left|\begin{array}{cccc}
              \pi_L\left(-\omega_1\right) & \ldots & \pi_L\left(-\omega_N\right) & (-z)^L \\
              \pi_{L+1}\left(-\omega_1\right) & \ldots & \pi_{L+1}\left(-\omega_N\right) & (-z)^{L+1} \\
              \vdots & \ddots & \vdots & \vdots \\
              \pi_{L+N}\left(-\omega_1\right) & \ldots & \pi_{L+N}\left(-\omega_N\right) & (-z)^{L+N} \\
            \end{array}
\right|}{\left|\begin{array}{ccc}
              \pi_L\left(-\omega_1\right) & \ldots & \pi_L\left(-\omega_N\right) \\
              \pi_{L+1}\left(-\omega_1\right) & \ldots & \pi_{L+1}\left(-\omega_N\right) \\
              \vdots & \ddots &  \vdots  \\
             \pi_{L+N-1}\left(-\omega_1\right) & \ldots & \pi_{L+N-1}\left(-\omega_N\right) \\
            \end{array}\right|}.
\end{split}
\end{equation}
We know that
\begin{equation}\label{RepresentationOfzn}
z^n=n!e^{z}\int\limits_0^{\infty}L_n(y)e^{-y}{}_0F_1(1;-zy)dy,
\end{equation}
see, for example, equation (2.12) in Forrester and Liu \cite{ForresterLiu}.
For our purposes it is convenient to rewrite equation (\ref{RepresentationOfzn}) as
\begin{equation}\label{RepresentationOfzn1}
(-z)^n=e^{z}\int\limits_0^{\infty}\pi_n(y)e^{-y}{}_0F_1 \left (1;-zy \right ) dy.
\end{equation}
Inserting formula (\ref{RepresentationOfzn1}) into equation (\ref{A1}), we find
\begin{equation}\label{A2}
\begin{split}
\E\left[\prod\limits_{i=1}^N\left(z-x_i\right)\right]=\frac{(-1)^{L+N}e^{z}}{z^L}\int\limits_0^{\infty}e^{-y}{}_0F_1(1;-zy)
\frac{\left|\begin{array}{cccc}
              \pi_L\left(-\omega_1\right) & \ldots & \pi_L\left(-\omega_N\right) & \pi_L(y) \\
              \pi_{L+1}\left(-\omega_1\right) & \ldots & \pi_{L+1}\left(-\omega_N\right) & \pi_{L+1}(y) \\
              \vdots & \ddots & \vdots & \vdots \\
              \pi_{L+N}\left(-\omega_1\right) & \ldots & \pi_{L+N}\left(-\omega_N\right) & \pi_{L+N}(y) \\
            \end{array}
\right|}{\left|\begin{array}{ccc}
              \pi_L\left(-\omega_1\right) & \ldots & \pi_L\left(-\omega_N\right)  \\
              \pi_{L+1}\left(-\omega_1\right) & \ldots & \pi_{L+1}\left(-\omega_N\right)  \\
              \vdots & \ddots & \vdots\\
              \pi_{L+N-1}\left(-\omega_1\right) & \ldots & \pi_{L+N-1}\left(-\omega_N\right)  \\
            \end{array}\right|}dy.
\end{split}
\end{equation}
It remains to compute the ratio of determinants in the formula just written above in terms of multiple integrals.
We use Proposition \ref{PropositionBrezinHikami} to rewrite both the numerator (with $n=L$, $m=N+1$, $\mu_1=-\omega_1$, $\ldots$,
$\mu_N=-\omega_N$, $\mu_{N+1}=y$) and the denominator (with $n=L$, $m=N+1$, $\mu_1=-\omega_1$, $\ldots$,
$\mu_N=-\omega_N$, $\mu_{N+1}=y$) as
\begin{equation}
\begin{split}
&\frac{\triangle\left(-\omega_1,...,-\omega_{N},y\right)}{\triangle\left(-\omega_1,...,-\omega_{N}\right)} \frac{\int\limits_{0}^{\infty}...\int\limits_{0}^{\infty}\prod\limits_{i=1}^L\prod\limits_{j=1}^N
\left(-\omega_j-t_i\right)\prod\limits_{i=1}^L\left(y-t_i\right)\triangle^2\left(t_1,...,t_L\right)
\prod\limits_{i=1}^L e^{-t_i}dt_i }{\int\limits_{0}^{\infty}...\int\limits_{0}^{\infty}\prod\limits_{i=1}^L\prod\limits_{j=1}^N
\left(-\omega_j-t_i\right)\triangle^2\left(t_1,...,t_L\right)
\prod\limits_{i=1}^L e^{-t_i}dt_i },
\end{split}
\end{equation}
The ratio of the Vandermonde determinants can be written as
\begin{equation}
\frac{\triangle\left(-\omega_1,...,-\omega_{N},y\right)}{\triangle\left(-\omega_1,...,-\omega_{N}\right)}=(-1)^N \prod\limits_{j=1}^N\left(-\omega_j-y\right)
=\prod\limits_{j=1}^N\left(\omega_j+y\right),
\end{equation}
and we finally arrive at the following formula for the averaged characteristic polynomial:
\begin{equation}\label{A3}
\begin{split}
&\E\left[\prod\limits_{i=1}^N\left(z-x_i\right)\right]=(-1)^{N+L}\frac{e^{z}}{z^L}\int\limits_0^{\infty}dye^{-y}{}_0F_1(1;-zy)\prod_{j=1}^N\left(\omega_j+y\right) \\
&\times\frac{\int\limits_{0}^{\infty}...\int\limits_{0}^{\infty}\prod\limits_{i=1}^L\prod\limits_{j=1}^N
\left(\omega_j+t_i\right)\prod\limits_{i=1}^L\left(y-t_i\right)\triangle^2\left(t_1,...,t_L\right)
\prod\limits_{i=1}^L e^{-t_i}dt_i }{\int\limits_{0}^{\infty}...\int\limits_{0}^{\infty}\prod\limits_{i=1}^L\prod\limits_{j=1}^N
\left(\omega_j+t_i\right)\triangle^2\left(t_1,...,t_L\right)
\prod\limits_{i=1}^L e^{-t_i}dt_i}.
\end{split}
\end{equation}
\end{proof}
\section{Average of ratios of characteristic polynomials}
\label{RCPSection}
\begin{thm}\label{RatioCPAveragingTheorem} Consider the probability distribution on the space of complex matrices of size $N\times N$ defined by equation (\ref{distribution}).\\
\textbf{(A)} Let $\omega_1$, $\ldots$, $\omega_N$ be the squared singular values of $\Omega$, and assume that $L\geq 1$. We have
\begin{equation}
\begin{split}
&\E\left[\prod\limits_{i=1}^N\frac{v-x_i}{z-x_i}\right]=(-1)^{N+L+1} \frac{e^{v} v^{-L}}{\tilde{\mathcal{N}}_L} \frac{1}{2\pi i}\int\limits_0^{\infty}dx\frac{v-x}{z-x}x^L e^{-x} \int\limits_{C}\frac{du\;{}_0F_1(1;ux)e^{-u}}{\prod\limits_{j=1}^N\left(u-\omega_j\right)}\\
&\times\int\limits_0^{\infty}\frac{dse^{-s}{}_0F_1(1;-vs)}{s+u}\prod\limits_{j=1}^N\left(s+\omega_j\right) \int\limits_{0}^{\infty}...\int\limits_{0}^{\infty}\prod_{k=1}^L\prod_{j=1}^N
\left(\omega_j+t_k\right) \triangle^2\left(t\right)
\prod_{i=1}^L \frac{s-t_i}{u+t_i} e^{-t_i} dt_i,
\end{split}
 \end{equation}
 where $C$ is a counter-clockwise contour encircling the points $\omega_1$, $\ldots$, $\omega_N$, and leaving
 the real negative numbers $-s,-t_1,...,-t_L$ outside. The normalization constant $\tilde{\mathcal{N}}_L$ is given by the formula \eqref{tildeNormL}.\\
\textbf{(B)} Assume that $\Omega=z \mathbf{1}_N$, $z\in\C$. Then $\Omega^*\Omega=|z|^2\mathbf{1}_N$, and $\omega_1=\ldots=\omega_N=|z|^2$.
In this case we obtain the formula
\begin{equation}
\begin{split}
&\E\left[\prod\limits_{i=1}^N\frac{v-x_i}{z-x_i}\right]_{\text{deg}}=(-1)^{N+L+1} \frac{e^{v} v^{-L}}{\tilde{\mathcal{N}}_L} \frac{1}{2\pi i}\int\limits_0^{\infty}dx\frac{v-x}{z-x}x^L e^{-x} \int\limits_{C}\frac{du\;{}_0F_1(1;ux)e^{-u}}{\left(u-|z|^2\right)^N} \\
&\times \int\limits_0^{\infty}\frac{dse^{-s}{}_0F_1(1;-vs)}{s+u} \left(s+|z|^2\right)^N \int\limits_{0}^{\infty}...\int\limits_{0}^{\infty}\prod_{k=1}^L
\left(|z|^2+t_k\right)^N \triangle^2\left(t\right)
\prod_{i=1}^L \frac{s-t_i}{u+t_i} e^{-t_i} dt_i,
\end{split}
\end{equation}
where the normalization $\tilde{\mathcal{N}}_L = \int_0^{\infty}...\int_0^{\infty}\prod_{i=1}^L \left(t_i+|z|^2\right)^N \triangle^2(t) \prod_{i=1}^L e^{-t_i}dt_i$.
\end{thm}
\begin{proof}
We begin from the following Proposition
\begin{prop}
With the same notation as in the statement of Proposition \ref{PropositionGeneralFormula} the formula for
the average of a ration of two  characteristic polynomials can be written as
\begin{equation}
\begin{split}
&\E\left[\prod\limits_{i=1}^N\frac{v-x_i}{z-x_i}\right]=\int\limits_{0}^{\infty}dx\frac{v-x}{z-x} \\
&\times \sum\limits_{i=1}^N\zeta_i(x)
\frac{\left|\begin{array}{ccccccc}
  g_{1,1} & \ldots & g_{1,i-1}&  1& g_{1,i+1} &\ldots & g_{1,N}  \\
  g_{2,1} & \ldots & g_{2,i-1}& v & g_{2,i+1} &\ldots & g_{2,N}  \\
  \vdots & \ddots & \vdots & \vdots & \vdots & \ddots & \vdots \\
  g_{N,1} & \ldots & g_{N,i-1}& v^{N-1}  & g_{N,i+1} &\ldots & g_{N,N}  \\
\end{array}\right|}{\left|\begin{array}{ccc}
  g_{1,1} & \ldots & g_{1,N} \\
  \vdots & \ddots & \vdots \\
  g_{N,1} & \ldots & g_{N,N}
\end{array}\right|}.
\end{split}
\end{equation}
\end{prop}
\begin{proof} This formula can be derived by the procedure very similar to that in Desrosiers and  Forrester \cite{DesrosiersForrester},
see Proposition 2.
\end{proof}
The ratio of determinants in the formula above can be rewritten as
\begin{equation}
\begin{split}
\frac{e^{-\omega_i}}{v^L} \frac{\left|\begin{array}{ccccccc}
\pi_L\left(-\omega_1\right) & \ldots & \pi_L\left(-\omega_{i-1}\right) & (-v)^L &\pi_L\left(-\omega_{i+1}\right) & \ldots &\pi_L\left(-\omega_N\right) \\
\pi_{L+1}\left(-\omega_1\right) & \ldots & \pi_{L+1}\left(-\omega_{i-1}\right) & (-v)^{L+1} &\pi_{L+1}\left(-\omega_{i+1}\right) & \ldots &\pi_{L+1}\left(-\omega_N\right) \\              \vdots & \ddots & \vdots & \vdots & \vdots & \ddots & \vdots\\
\pi_{L+N-1}\left(-\omega_1\right) & \ldots & \pi_{L+N-1}\left(-\omega_{i-1}\right) & (-v)^{L+N-1} &\pi_{L+N-1}\left(-\omega_{i+1}\right) & \ldots &\pi_{L+N-1}\left(-\omega_N\right) \\              \end{array}
\right|}{\left|\begin{array}{ccc}
              \pi_L\left(-\omega_1\right) & \ldots & \pi_L\left(-\omega_N\right)  \\
              \pi_{L+1}\left(-\omega_1\right) & \ldots & \pi_{L+1}\left(-\omega_N\right)  \\
              \vdots & \ddots & \vdots \\
              \pi_{L+N-1}\left(-\omega_1\right) & \ldots & \pi_{L+N-1}\left(-\omega_N\right)  \\
            \end{array}\right|}.
\end{split}
\nonumber
\end{equation}
We use formula (\ref{RepresentationOfzn1}) to rewrite the equation just written above as
\begin{equation}
\begin{split}
& \frac{e^{-\omega_i+v}}{v^L}\int\limits_0^{\infty}dse^{-s}{}_0F_1(1;-vs) \\
&\times\frac{\left|\begin{array}{ccccccc}
\pi_L\left(-\omega_1\right) & \ldots & \pi_L\left(-\omega_{i-1}\right) & \pi_L(s) &\pi_L\left(-\omega_{i+1}\right) & \ldots &\pi_L\left(-\omega_N\right) \\
\pi_{L+1}\left(-\omega_1\right) & \ldots & \pi_{L+1}\left(-\omega_{i-1}\right) & \pi_{L+1}(s) &\pi_{L+1}\left(-\omega_{i+1}\right) & \ldots &\pi_{L+1}\left(-\omega_N\right) \\              \vdots & \ddots & \vdots & \vdots & \vdots & \ddots & \vdots \\
\pi_{L+N-1}\left(-\omega_1\right) & \ldots & \pi_{L+N-1}\left(-\omega_{i-1}\right) & \pi_{L+N-1}(s) &\pi_{L+N-1}\left(-\omega_{i+1}\right) & \ldots &\pi_{L+N-1}\left(-\omega_N\right) \\              \end{array}
\right|}{\left|\begin{array}{ccc}
              \pi_L\left(-\omega_1\right) & \ldots & \pi_L\left(-\omega_N\right)  \\
              \pi_{L+1}\left(-\omega_1\right) & \ldots & \pi_{L+1}\left(-\omega_N\right)  \\
              \vdots & \ddots  & \vdots \\
              \pi_{L+N-1}\left(-\omega_1\right) & \ldots & \pi_{L+N-1}\left(-\omega_N\right)  \\
            \end{array}\right|}.
\end{split}
\nonumber
\end{equation}
Again, we represent the ratio of determinants with orthogonal polynomial entries in terms of multiple integrals using Proposition \ref{PropositionBrezinHikami} to both numerator and denominator:
\begin{equation}
\begin{split}
&\frac{\triangle\left(-\omega_1,...,-\omega_{i-1},s,-\omega_{i+1},...,-\omega_N\right)}{\triangle\left(-\omega_1,...,-\omega_N\right)} \frac{\int\limits_{0}^{\infty}...\int\limits_{0}^{\infty}\prod\limits_{k=1}^L\prod\limits_{\substack{j=1\\j\neq i}}^N
\left(-\omega_j-t_k\right)\prod\limits_{k=1}^L\left(s-t_k\right)\triangle^2\left(t_1,...,t_L\right)
\prod\limits_{i=1}^L e^{-t_i} dt_i}{\int\limits_{0}^{\infty}...\int\limits_{0}^{\infty}\prod\limits_{k=1}^L\prod\limits_{j=1}^N
\left(-\omega_j-t_k\right)\triangle^2\left(t_1,...,t_L\right)
\prod\limits_{i=1}^L e^{-t_i} dt_i}.
\end{split}
\end{equation}
The ratio of the Vandermonde determinants can be simplified as
\begin{equation}
\frac{\triangle\left(-\omega_1,\ldots,-\omega_{i-1},s,-\omega_{i+1},\ldots,-\omega_N\right)}{\triangle\left(-\omega_1,\ldots,-\omega_N\right)}
=\underset{k\neq i}{\prod\limits_{k=1}^N}\frac{\omega_k+s}{\omega_k-\omega_i}.
\end{equation}
Therefore, the involved ratio of the determinants can be rewritten as
\begin{equation}
\begin{split}
(-1)^L\prod\limits_{\substack{k=1\\k\neq i}}^N\frac{\omega_k+s}{\omega_k-\omega_i}
\frac{\int\limits_{0}^{\infty}...\int\limits_{0}^{\infty}\prod\limits_{k=1}^L\prod\limits_{\substack{j=1\\j\neq i}}^N
\left(\omega_j+t_k\right)\prod\limits_{k=1}^L\left(s-t_k\right)\triangle^2\left(t_1,...,t_L\right)
\prod\limits_{i=1}^L e^{-t_i} dt_i}{\int\limits_{0}^{\infty}...\int\limits_{0}^{\infty}\prod\limits_{k=1}^L\prod\limits_{j=1}^N
\left(\omega_j+t_k\right)\triangle^2\left(t_1,...,t_L\right)
\prod\limits_{i=1}^L e^{-t_i} dt_i}.
\end{split}
\nonumber
\end{equation}
These calculations give us the following formula for the average of ratios of characteristic polynomials
\begin{equation}
 \begin{split}
&\E\left[\prod\limits_{i=1}^N\frac{v-x_i}{z-x_i}\right]=\int\limits_0^{\infty}dx\frac{v-x}{z-x}\\
&\times\sum\limits_{i=1}^N\zeta_i(x)
\left(\frac{e^{-\omega_i+v}}{v^L}\int\limits_0^{\infty}dse^{-s}{}_0F_1(1;-vs)(-1)^L\prod\limits_{\substack{k=1\\k\neq i}}^N
\frac{\omega_k+s}{\omega_k-\omega_i}\right)\\
&\times\frac{\int\limits_{0}^{\infty}...\int\limits_{0}^{\infty}\prod\limits_{k=1}^L\prod\limits_{\substack{j=1\\j\neq i}}^N
\left(\omega_j+t_k\right)\prod\limits_{k=1}^L\left(s-t_k\right)\triangle^2\left(t_1,...,t_L\right)
\prod\limits_{i=1}^L e^{-t_i} dt_i}{\int\limits_{0}^{\infty}...\int\limits_{0}^{\infty}\prod\limits_{k=1}^L\prod\limits_{j=1}^N
\left(\omega_j+t_k\right)\triangle^2\left(t_1,...,t_L\right)
\prod\limits_{i=1}^L e^{-t_i} dt_i}.
\end{split}
 \end{equation}
 Taking into account the explicit form of the functions $\zeta_i(x)$ (see equation (\ref{zetaeta})), and applying the Basic Residue Theorem, we obtain
 the formula in the statement of Theorem \ref{RatioCPAveragingTheorem} after some straightforward manipulations.
\end{proof}
\section{The formula for the correlation kernel}
\label{KCPSection}
\begin{thm}
\label{KernelAveragingTheorem}
Consider the determinantal process formed by the squared singular values of  a random matrix $X$
whose probability distribution is defined by formula (\ref{distribution}). The correlation kernel $K_N(x,y)$
for this determinantal point process can be written as
\begin{equation}
\label{kernel}
\begin{split}
K_N(x,y) = & \frac{(-1)^{N+L+1}e^{x-y}}{\tilde{\mathcal{N}}_L} \left ( \frac{y}{x} \right )^L \frac{1}{2\pi i} \int\limits_{C}\frac{du\;{}_0F_1(1;uy)e^{-u}}{\prod\limits_{j=1}^N\left(u-\omega_j\right)}\int\limits_0^{\infty}\frac{dse^{-s}{}_0F_1(1;-xs)}{s+u}\prod\limits_{j=1}^N\left(s+\omega_j\right) \\
& \times \int\limits_{0}^{\infty}...\int\limits_{0}^{\infty}\prod_{k=1}^L\prod_{j=1}^N
\left(\omega_j+t_k\right) \triangle^2\left(t\right)
\prod_{i=1}^L \frac{s-t_i}{u+t_i} e^{-t_i} dt_i,
\end{split}
\end{equation}
 where $C$ is a counter-clockwise contour encircling the points $\omega_1$, $\ldots$, $\omega_N$, and leaving
 the real negative numbers $-s,-t_1,...,-t_L$ outside. The normalization constant $\tilde{\mathcal{N}}_L$ is given by formula \eqref{tildeNormL}.
\end{thm}
\begin{proof} It is known that for polynomial ensembles the correlation kernel can be written in terms of
averages of ratios of characteristic polynomials as follows
\begin{equation}
K_N(x,y)=\frac{1}{x-y}\underset{z=y}{\Res}\left(\E\left[\prod\limits_{i=1}^N\frac{x-x_i}{z-x_i}\right]\right),
\end{equation}
see Desrosiers and  Forrester \cite{DesrosiersForrester}, equation (11). Use the equation just written above, and Theorem \ref{RatioCPAveragingTheorem}.
\end{proof}

\section{Joint bulk-edge regime asymptotics in the degenerate case}
\label{AsymptoticsSection}
We study the $\omega_1 = ... = \omega_N = |z|^2$ degenerate case of averages considered in Theorems \ref{InverseCPAveragingTheorem}B, \ref{CPAveragingTheorem}B and \ref{RatioCPAveragingTheorem}B in a \emph{joint bulk-edge} asymptotic regime where $|z|^2 \sim N$ (bulk regime from the point of view of eigenvalues of the matrix $X$) and the argument of characteristic polynomials scales as $N^{-1}$ (edge regime from the point of view of eigenvalues of the matrix $X^\dagger X$). Such choice is motivated, in particular, by the applications \cite{FyodorovEigenvectors} and a detailed discussion of the joint bulk-edge regime is given in Sec. \ref{connection}.
\subsection{Asymptotics of the averaged inverse characteristic polynomial}
\begin{prop}
\label{ICPAsymptotics}
In the degenerate case $\omega_1 = ... = \omega_N = |z|^2$, the averaged inverse characteristic polynomial $\E\left[\prod\limits_{i=1}^N\frac{1}{p+x_i}\right]_{\text{deg}} \equiv Q_{N}(p)$ given by equation \eqref{ICPDegForm} has the following joint bulk-edge asymptotic behaviour for general $L\geq 0$
\begin{equation}
\begin{split}
&\lim_{N\to \infty} \frac{N^{N+L-1/2}}{e^{NR_*}} Q_N\left (p = \frac{\xi}{NR_*} \right ) = \sqrt{\frac{2}{\pi}} \xi^{\frac{L}{2}} K_L\left (2\sqrt{\xi} \right ),
\end{split}
\end{equation}
where $R_* = 1-R$ and the source is scaled with $N$ as $|z|^2 = NR$.
\end{prop}
\begin{proof}
We proceed similarly to the proof of Proposition \ref{L0Equivalence}. The formula \eqref{ICPDegForm}
is rewritten by using the expression $\frac{1}{p+u} = \int_0^1 \frac{d\tau}{(1-\tau)^2} e^{-\frac{(p+u)\tau}{1-\tau}}$ and the identity \eqref{Integral} to evaluate the integral over $u$. The result reads (cf. Eq.(\ref{Gen3mynorm}))
\begin{equation}
\begin{split}
\label{qGrel}
&Q_N(p) = \int\limits_0^{1} d\tau (1-\tau)^{L-1} e^{-\frac{\tau p}{1-\tau}}\, G^{(L)}_N\left(|z|^2,\tau\right),
\end{split}
\end{equation}
where
\begin{equation*}
\begin{split}
 G^{(L)}_N\left(|z|^2,\tau\right) = & \frac{(-1)^{N-1} L!}{\tilde{\mathcal{N}}_L}\,\int\limits_0^{\infty}... \int\limits_0^{\infty}\triangle^2(t) \prod_{k=1}^L (t_k+|z|^2)^N\,e^{-t_k} dt_k \frac{1}{2\pi i} \int\limits_{C} \frac{dv e^{-v\,\tau}}{(v-|z|^2)^N}\frac{L_L\left(v(\tau-1)\right)}{\prod\limits_{k=1}^L (t_k+v)}.
\end{split}
\end{equation*}
To extract the scaling limit of $G^{(L)}_N\left(|z|^2,\tau\right)$, we rescale the parameters $|z|^2= NR$ and $\tau=1-a/N$ and integration variables $v\to Nv$ and $t_k= Nq_k$. In this way we define the rescaled function $g(R,a) = \frac{G^{(L)}_N\left (NR,1-\frac{a}{N} \right )}{N^{(N+L)(L-1)+1}}$ being equal to
\begin{equation}\label{XGen4}
g(R,a) = \frac{(-1)^{N-1} L!}{\tilde{\mathcal{N}}_L}\int\limits_0^{\infty}... \int\limits_0^{\infty}\triangle^2(q) \prod_{k=1}^L (q_k+R)^N\,e^{-Nq_k} dq_k \frac{1}{2\pi i} \int\limits_{C} \frac{dv e^{-Nv}}{(v-R)^N}\frac{e^{va}L_L\left(-va\right)}{\prod\limits_{k=1}^L (q_k+v)}.
\end{equation}
We use an identity
\begin{equation}\label{Xident1L}
\prod_{k=1}^{L}\frac{1}{q_k+v}=\sum_{k=1}^L\frac{1}{q_k+v}\prod_{\substack{m=1\\m\ne k}}^L\frac{1}{(q_m-q_k)}
\end{equation}
which, when substituted back to equation \eqref{XGen4}, reduces the sum over $k$ to its last term $\frac{1}{q_L+v}\prod_{m=1}^{L-1}\frac{1}{(q_m-q_L)}$ and multiply the result by $L$ due to the symmetry of $q_k$ integrals. The distinguished variable is renamed $q_L \to q$
\begin{equation}\label{XGen5}
\begin{split}
g(R,a) = & \frac{(-1)^{N+L-2} L L!}{\tilde{\mathcal{N}}_L}\int\limits_0^{\infty}... \int\limits_0^{\infty}\triangle^2(q_{1},...,q_{L-1}) \prod_{k=1}^{L-1} (q_k+R)^N\,e^{-Nq_k} dq_k \\
& \times \int\limits_0^\infty d q (q+R)^N e^{-Nq} \prod\limits_{i=1}^{L-1} (q - q_i) \frac{1}{2\pi i} \int\limits_{C} \frac{dv e^{-Nv}}{(v-R)^N} \frac{e^{va}L_L\left(-va\right)}{q + v}.
\end{split}
\end{equation}
We consider the contour integral
\begin{equation}
\label{contourint}
I(q) = \frac{1}{2\pi i} \int\limits_{C} dv e^{-N\mathcal{L}_1(v)} \frac{e^{va}L_L\left(-va\right)}{q + v},
\end{equation}
with $\mathcal{L}_1(v) = v + \log(v-R)$ which we approximate by the method of steepest descent. Denoting $R_* = 1-R$, the saddle point is found at $v_* = -R_*$ and located on the negative real axis as $R<1$. Since the orginal contour $C$ encircles the pole at $v=R$ and does not cross the negative real axis, it has to be deformed to pass through the saddle $v_*$. This adds a contribution to the integral coming from the pole at $v=-q$ as long as $-q>v_*$:
\begin{equation*}
\int\limits_{C} = \int\limits_{C_{\text{sp}}(R)} - \int\limits_{C(-q)} \theta(-q -v_*),
\end{equation*}
where $C_{\text{sp}}(R)$ is a contour encircling $R$ and passing through the saddle point $v_*$ whereas $C(-q)$ is a path encircling $-q$. We compute both contributions
\begin{align}
\label{approx}
I(q) = \frac{1}{2\pi i} \int\limits_{C_{\text{sp}}(R)} dv e^{-N\mathcal{L}_1(v)} \frac{e^{va}L_L\left(-va\right)}{q + v} + \theta(R_*-q) (-1)^{N+1} e^{N\mathcal{L}_2(q)-aq} L_L(aq)
\end{align}
and substitute it back into the equation \eqref{XGen5} to obtain:
\begin{equation*}\label{XGen6}
\begin{split}
g(R,a) \sim g_1 + g_2,
\end{split}
\end{equation*}
where
\begin{equation}
\label{g1g2def}
\begin{split}
g_1 = & \frac{(-1)^{N+L}L L! }{ \tilde{\mathcal{N}}_L}\int\limits_0^{\infty}... \int\limits_0^{\infty}\triangle^2(q_{1},...,q_{L-1}) \prod_{k=1}^{L-1} \,e^{-N\mathcal{L}_2(q_k)} dq_k \\
& \times \int\limits_0^{\infty} d q \prod\limits_{i=1}^{L-1} (q - q_i) e^{-N\mathcal{L}_2(q)} \frac{1}{2\pi i} \int\limits_{C_{\text{sp}}(R)} dv e^{-N\mathcal{L}_1(v)} \frac{e^{va}L_L\left(-va\right)}{q + v}, \\
g_2 = & \frac{(-1)^{L-1}L L!}{\tilde{\mathcal{N}}_L}\int\limits_0^{\infty}... \int\limits_0^{\infty}\triangle^2(q_{1},...,q_{L-1}) \prod_{k=1}^{L-1} \,e^{-N\mathcal{L}_2(q_k)} dq_k \int\limits_0^{R_*} d q \prod\limits_{i=1}^{L-1} (q - q_i) e^{-aq} L_L(aq)
\end{split}
\end{equation}
and $\mathcal{L}_2(q) = q - \ln (q+R)$. We first compute the $g_1$ term by approximating the $q_k$, $q$ integrals around the saddle points $q_k = R_* + N^{-1/2} x_k$, $q = R_* + N^{-1/2} x$ and the contour integral $v$ around $v = -R_* + iN^{-1/2} y$:
\begin{equation}
\begin{split}
g_1 \sim \frac{(-1)^{L}L L! e^{-N(L-1)R_*}}{N^{(L-1)^2/2} \tilde{\mathcal{N}}_L} \frac{ e^{-a R_*} L_L(aR_*)}{2\pi N^{L/2}} & \int\limits_{-\infty}^{\infty}... \int\limits_{-\infty}^{\infty}\triangle^2(x_{1},...,x_{L-1}) \prod_{k=1}^{L-1} \,e^{-\frac{x_k^2}{2}} dx_k \\
& \times \int\limits_{-\infty}^{\infty} d x dy \frac{e^{-\frac{x^2}{2} - \frac{y^2}{2}}}{x+iy}\prod\limits_{i=1}^{L-1} (x - x_i) .
\end{split}
\end{equation}
In the above formula, we introduce the following representation of the Hermite monic polynomials $h_n$, see e.g. Ref. \cite{FyoStra2003a}
\begin{equation}
\label{MonHerm}
\int\limits_{-\infty}^\infty \prod\limits_{i=1}^{L-1} e^{-\frac{x_i^2}{2}} dx_i \triangle^2(x) \prod\limits_{i=1}^{L-1} (z-x_i) = h_{L-1} (z) \prod_{k=1}^{L-1} \left ( k! \sqrt{2\pi} \right ),
\end{equation}
which results in the first contribution $g_1$ equal to
\begin{equation}
\label{g1}
g_1(R,a) \sim \frac{(-1)^{L}L L! e^{-N(L-1)R_*}}{N^{(L-1)^2/2} \tilde{\mathcal{N}}_L}  \prod_{i=1}^{L-1} (i! \sqrt{2\pi}) \frac{ e^{-a R_*} L_L(aR_*)}{2\pi N^{L/2}} \int\limits_{-\infty}^{\infty} d x dy \frac{e^{-\frac{x^2}{2} - \frac{y^2}{2}}}{x+iy} h_{L-1}(x),
\end{equation}
where the remaining $x,y$ integrals do not depend on $N$ any longer. We now turn our attention to the second term $g_2$ given by equation \eqref{g1g2def}. In this case, the $q$ integral is computed exactly due to an exact cancellation of the $N$-dependent parts:
\begin{equation}\label{g2}
\begin{split}
g_2(R,a) \sim & \frac{(-1)^{L-1}L L! e^{-N(L-1)R_*}}{N^{(L-1)^2/2}\tilde{\mathcal{N}}_L} \prod_{i=1}^{L-1} (i! \sqrt{2\pi}) \int\limits_0^{R_*} d q (q -R_*)^{L-1} e^{-aq} L_L(aq),
\end{split}
\end{equation}
and using a Mehta-type integral formula $\int\limits_{-\infty}^{\infty}\triangle^2(x_{1},...,x_{n}) \prod\limits_{k=1}^{n} \,e^{-\frac{x_k^2}{2}} dx_k = \prod\limits_{i=1}^{n} (i! \sqrt{2\pi})$.

By comparing both contributions to $g$ given by equations \eqref{g1} and \eqref{g2} we see that $g_2/g_1 \sim N^{L/2}$. Hence, the term $g_2$ is dominant in the large $N$ limit and we can safely disregard the contribution $g_1$ in what follows. The remaining $q$ integral in the contribution coming from \eqref{g2} is found exactly by expressing Laguerre polynomial as $e^{-x} L_n(x) = \frac{1}{n!} \frac{d^n}{dx^n} \left (e^{-x} x^n \right )$ and integrating by parts:
\begin{equation}
\label{qint}
\int\limits_0^{R_*} d q (q -R_*)^{L-1} e^{-aq} L_L(aq) = - \frac{1}{L} e^{-aR_*} (-R_*)^L .
\end{equation}
In a similar fashion we approximate the normalization constant $\tilde{\mathcal{N}}_L$ given by formula \eqref{tildeNormL} as
\begin{align}
\label{nlapprox}
\tilde{\mathcal{N}}_L \sim N^{NL+L^2/2} e^{-NLR_*} \prod\limits_{i=1}^{L} (i! \sqrt{2\pi}).
\end{align}
At last, by using equations \eqref{qint} and \eqref{nlapprox} in Eq.\eqref{XGen6} we arrive at
\begin{equation}\label{XGen8}
\begin{split}
G^{(L)}_N\left (R,1-\frac{a}{N} \right ) \sim & \frac{e^{NR_*}}{N^{N-1/2}} \frac{e^{-aR_*}R_*^L}{\sqrt{2\pi} }.
\end{split}
\end{equation}
Such an expression and the integral representation of the Bessel-Macdonald function $\int_0^\infty da a^{L-1} e^{-\frac{\xi}{a} - aR_*} = 2 \left ( \frac{\xi}{R_*} \right )^{\frac{L}{2}} K_L \left (2\sqrt{\xi R_*} \right )$
can now be used in equation \eqref{qGrel} to obtain our final asymptotic formula for the 'joint bulk-edge asymptotics' of the mean of inverse characteristic polynomial:
\begin{equation*}
Q_N\left (p=\frac{\xi}{NR_*} \right ) \sim \sqrt{\frac{2}{\pi}}\frac{e^{NR_*}}{N^{N+L-1/2}} \xi^{\frac{L}{2}} K_L\left (2\sqrt{\xi} \right ).
\end{equation*}
equivalent to the statement of the Proposition.
\end{proof}

The above considerations can be trivially translated into the asymptotic behavior of the object $\mathcal{G}^{(L)}_N(\rho,\tau)$ defined in (\ref{Gen2mynorm}),(\ref{Gen3mynorm})
and featuring in applications in \cite{FyodorovEigenvectors}. Namely, we have after rescaling $\rho=Nw^2$ and exploiting (\ref{XGen8})
\begin{equation}\label{GLafinmynorm}
\lim_{N\to \infty}\frac{1}{N^{(N-1/2)(L-1)+1}}\mathcal{G}^{(L)}_N(Nw^2,\tau=1-a/N)= \left(\prod_{k=1}^L\,k!\right) (2\pi)^{(L-1)/2} e^{-a(1-w^2)} (1-w^2)^L
\end{equation}
By renaming $a\to 1/s$ and considering the special case $L=2$ one can check that this formula is exactly equivalent, after appropriate interpretation, to the expression Eq.(2.24) in  \cite{FyodorovEigenvectors} derived there via a relatively tedious procedure.

\subsection{Asymptotics of the averaged characteristic polynomial}
\begin{prop}
\label{CPAsymptotics}
In the degenerate case $\omega_1 = ... = \omega_N = |z|^2$, the averaged characteristic polynomial $\E\left[\prod\limits_{i=1}^N(p+x_i)\right]_{\text{deg}} \equiv P_{N}(p)$ given by equation \eqref{CPDegForm} has the following joint bulk-edge asymptotic behaviour for general $L\geq 0$
\begin{equation}
\begin{split}
\lim_{N\to \infty} \frac{e^{NR_*}}{N^{L + N +1/2} } P_N\left (p = \frac{\xi}{NR_*} \right ) = (-1)^L \sqrt{2\pi} \xi^{-\frac{L}{2}} I_L \left ( 2 \sqrt{ \xi } \right ),
\end{split}
\end{equation}
where $R_* = 1-R$, the source scaled as $|z|^2 = NR$ and $I_L(x)$ stands for the Bessel function of imaginary argument.
\end{prop}
\begin{proof}
We start with the formula \eqref{CPDegForm} in which we rescale the parameters accordingly $|z|^2 = NR$, $y\to Ny$, $p = \xi/N$, $t_i = N q_i$ to find
\begin{equation*}
\begin{split}
\frac{P_N \left ( p = \frac{\xi}{N} \right )}{N^{(N+L+1)(L+1)}} = & \frac{1}{\tilde{\mathcal{N}}_L}\frac{e^{-\xi/N}}{\xi^L}\int\limits_0^{\infty}dye^{-N\mathcal{L}_2(y)}I_0(2\sqrt{\xi y}) \int\limits_{0}^{\infty}...\int\limits_{0}^{\infty} \triangle^2\left(q\right)
\prod_{i=1}^L (y-q_i) e^{-N\mathcal{L}_2(q_i)}dq_i,
\end{split}
\end{equation*}
with $\mathcal{L}_2(y) = y - \ln (R+y)$. The integral is evaluated in the leading approximation via the saddle-point method by setting $q_i = R_* + N^{-1/2} x_i$, $y = R_* + N^{-1/2} z$ where $R_* = 1-R$ and using an representation of Hermite monic polynomials introduced in formula \eqref{MonHerm}. With the approximate normalization coefficient given by formula \eqref{nlapprox}, the result reads
\begin{equation*}
P_N \left ( p = \frac{\xi}{N} \right ) \sim \frac{N^{1/2 + \frac{3}{2}L +N} e^{-NR_*} }{\xi^{L}} \int\limits_{-\infty}^\infty dz e^{-\frac{z^2}{2}} I_0 \left ( 2 \sqrt{\xi R_* + N^{-1/2} z \xi} \right )  h_L(z).
\end{equation*}
By using the Rodrigues identity $h_L(z) = (-1)^L e^{z^2/2} \frac{d^L}{dz^L} e^{-z^2/2}$, the last integral is evaluated as
\begin{equation*}
\begin{split}
& \int_{-\infty}^\infty dz e^{-\frac{z^2}{2}} I_0 \left ( 2 \sqrt{\xi R_* + N^{-1/2} z \xi} \right ) h_L (z) = \int_{-\infty}^\infty dz e^{-z^2/2} \frac{d^L}{dz^L} I_0 \left ( 2 \sqrt{ \xi R_* + N^{-1/2} z \xi} \right ) \sim \\
& \sim (-1)^L \sqrt{2\pi} N^{-L/2} \frac{d^L}{dR^L} I_0 \left ( 2 \sqrt{ \xi R_*} \right ) = (-1)^L \sqrt{2\pi} N^{-L/2}  \left ( \frac{\xi}{R_*} \right )^{L/2} I_L \left ( 2 \sqrt{ \xi R_*} \right ).
\end{split}
\end{equation*}
and plugged into previous expression to obtain
\begin{equation*}
\begin{split}
P_N \left ( p = \frac{\xi}{N} \right ) \sim (-1)^L \frac{N^{1/2 + L +N} \sqrt{2\pi}}{e^{NR_*}} \left ( \frac{1}{\xi R_*} \right )^{L/2} I_L \left ( 2 \sqrt{ \xi R_*} \right ),
\end{split}
\end{equation*}
giving the final formula by rescaling $\xi \to \xi/R_*$.
\end{proof}
\subsection{Asymptotics of the kernel}
\begin{prop}
\label{RATIOAsymptotics}
In the degenerate case $\omega_1 = ... = \omega_N = |z|^2$, the kernel $K_{N}(x,y)$ given by equation \eqref{kernel} has the following joint bulk-edge asymptotic behaviour for general $L\geq 0$
\begin{equation}
\label{mainkernelformula}
\begin{split}
& \lim_{N\to \infty} \frac{1}{NR_*} K_N \left (x = \frac{\alpha}{NR_*},y = \frac{\beta}{NR_*} \right ) = \left ( \frac{\beta}{\alpha} \right )^{\frac{L}{2}} \int\limits_0^1 d\tau J_L\left (2\sqrt{\alpha \tau} \right ) J_L\left (2\sqrt{\beta \tau} \right ),
\end{split}
\end{equation}
where $R_* = 1-R$, the source is scaled as $|z|^2 = NR$  and $J_L(x)$ stands for the standard Bessel function.
\end{prop}
\begin{proof}
The formula for the kernel in the degenerate case is given by expression \eqref{kernel} upon setting $\omega_1 = ... = \omega_N = |z|^2$. We first rescale variables in that equation $|z|^2 = NR$, $t_i \to N t_i$, $u \to Nu$, $s \to Ns$ and introduce a rescaled kernel
\begin{equation}
k(\xi,\eta) = \frac{K_N\left (x = \frac{\xi}{N},y = \frac{\eta}{N} \right )}{N^{L(N+L) + 1}}
\end{equation}
Taking
into account the explicit formula for the kernel $K_N(x,y)$ obtained in Theorem \ref{KernelAveragingTheorem} we
can represent $k(\xi,\eta)$ as
\begin{equation}
\label{kdef}
\begin{split}
& k(\xi,\eta) = \frac{(-1)^{N+L+1}e^{\frac{\xi - \eta}{N}} \left (\frac{\eta}{\xi} \right )^L}{\tilde{\mathcal{N}}_L} \frac{1}{2\pi i} \int\limits_{C}\frac{du\;{}_0F_1(1;u\eta)e^{-Nu}}{\left(u-R\right)^N} \\
&\times \int\limits_0^{\infty}\frac{dse^{-Ns}{}_0F_1(1;-\xi s)}{s+u}\left(s+R\right)^N \int\limits_{0}^{\infty}...\int\limits_{0}^{\infty}\prod_{k=1}^L
\left(R+t_k\right)^N \triangle^2\left(t\right)
\prod_{i=1}^L \frac{s-t_i}{u+t_i} e^{-Nt_i} dt_i,
\end{split}
\end{equation}
Similarly as in the proof of Proposition \ref{ICPAsymptotics}, the term $\prod_{i=1}^L \frac{1}{u+t_i}$ of the above formula is expressed by the identity \eqref{Xident1L}, the last integration variable is singled out and renamed $t_L=t$
\begin{equation}
\label{maink}
k(\xi,\eta) = \frac{L(-1)^{N}e^{\frac{\xi - \eta}{N}} \left (\frac{\eta}{\xi} \right )^L}{\tilde{\mathcal{N}}_L} \int\limits_0^\infty dt_1... \int\limits_0^\infty dt_{L-1} \int\limits_0^\infty ds \int\limits_0^\infty dt A_2, \int\limits_{C} du A_1
\end{equation}
where the integrands are equal to
\begin{equation}
\begin{split}
& A_1 = \frac{1}{2\pi i} e^{-N\mathcal{L}_1(u)} {}_0F_1(1;u\eta) \frac{s-t}{(u+t)(u+s)}, \\
& A_2 = e^{-N\mathcal{L}_2(t)-N\mathcal{L}_2(s)} \prod_{k=1}^{L-1} \left [ e^{-N\mathcal{L}_2(t_k)} (s-t_k) (t-t_k) \right ]  {}_0 F_1(1;-\xi s) \triangle^2(t_1...t_{L-1}),
\end{split}
\end{equation}
with notation $\mathcal{L}_1(x) = x + \log(x-R), \mathcal{L}_2(x) = x - \ln (R+x)$ introduced already in the proofs of Propositions \ref{ICPAsymptotics} and \ref{CPAsymptotics}. Both the order of integrations and grouping of integrands are selected with the view to facilitate the subsequent asymptotic analysis.
The first integral in equation \eqref{maink}
\begin{equation}
\int\limits_C du A_1 = I_C(t) - I_C(s)
\end{equation}
is, by the formula $\frac{s-t}{(u+t)(u+s)} = \frac{1}{u+t} - \frac{1}{u+s}$, equal to the difference between the values of the same contour integral $I_C$ evaluated at two different arguments. This fundamental object reads
\begin{equation*}
I_C(x) = \frac{1}{2\pi i} \int\limits_C du e^{-N\mathcal{L}_1(u)} \frac{{}_0F_1(1;u\eta)}{u+x}
\end{equation*}
and differs from the previously studied integral \eqref{contourint} only by the $N$-independent part. Then, following the same analysis as has been done in the proof of Proposition \ref{ICPAsymptotics}, the saddle-point method identifies the saddle  point at $u_* = -R_*$ where $R_* = 1-R$ and the leading-order approximation comprising two contributions is given by
\begin{equation*}
I_C(x) \sim \frac{1}{2\pi i} \int\limits_{C_{sp}(R)} du e^{-N\mathcal{L}_1(u)} \frac{{}_0F_1(1;u\eta)}{u+x} + (-1)^{N+1} \frac{e^{Nx}{}_0F_1(1;-x\eta)}{(x+R)^N}  \theta(R_*-x),
\end{equation*}
One can again argue that the first term yields eventually a sub-leading contribution, whereas $C_{sp}(R)$ denotes a contour encircling the pole at $R$ and passing through the saddle point $u_*$. Thus, the first integration yields
\begin{equation*}
\int\limits_C du A_1 \sim (-1)^{N+1} e^{N\mathcal{L}_2(t)}{}_0F_1(1;-t\eta)\theta(R_*-t) - (-1)^{N+1} e^{N\mathcal{L}_2(s)}{}_0F_1(1;-s\eta)\theta(R_*-s)
\end{equation*}
which in turn allows to represent the rescaled kernel \eqref{maink} as a sum of two parts
\begin{equation}
\label{kappa}
 k(\xi,\eta) \sim \frac{L}{\tilde{\mathcal{N}}_L} \left ( \frac{\eta}{\xi} \right )^L \left ( J_2 - J_1 \right ),
 \end{equation}
where we denoted
\begin{equation*}
\begin{split}
J_1 = \int\limits_0^\infty dt_1... \int\limits_0^\infty dt_{L-1} \int\limits_0^\infty ds \int\limits_0^{1-R} dt A_2^{(t)}, \qquad J_2 = \int\limits_0^\infty dt_1... \int\limits_0^\infty dt_{L-1} \int\limits_0^{1-R} ds \int\limits_0^\infty dt A_2^{(s)},
\end{split}
\end{equation*}
and rescaled the integrands as $A_2^{(t)} = A_2 e^{N\mathcal{L}_2(t)}{}_0F_1(1;-t\eta)$, $A_2^{(s)} = A_2 e^{N\mathcal{L}_2(s)}{}_0F_1(1;-s\eta)$.

\subsubsection{Asymptotics of $J_1$} We evaluate $J_1$ integral by the saddle-point procedure applied first to the $t_i$'s integrals by expanding $t_i = R_* + N^{-1/2} x_i$ around the corresponding saddle points $R_* = 1-R$:
\begin{equation}
\begin{split}
& J_1 \sim \frac{e^{-N(L-1)R_*}}{N^{\frac{L^2}{2} - \frac{1}{2}}} \int\limits_0^{R_*} dt ~{}_0F_1 \Big (1;-\eta t \Big ) \int\limits_{0}^\infty ds e^{-N\mathcal{L}_2(s)} {}_0 F_1 \Big (1,-\xi s \Big ) \\
& \times \int\limits_{-\infty}^\infty ... \int\limits_{-\infty}^\infty \prod_{i=1}^{L-1} e^{-\frac{x_i^2}{2}} dx_{i} \triangle^2(x_1,...,x_{L-1}) \prod_{k=1}^{L-1} \left (\sqrt{N}(s-R_*) - x_k \right ) \left (\sqrt{N}(t-R_*) - x_k \right ).
\end{split}
\end{equation}
The $L-1$ integrals can be then re-expressed in terms of the Hermite monic polynomials $h_n$ by the formula found in Ref. \cite{FyoStra2003a}:
\begin{equation}
\label{hermident}
\begin{split}
\int\limits_{-\infty}^\infty ... \int\limits_{-\infty}^\infty \prod_{i=1}^{L-1} e^{-\frac{x_i^2}{2}}dx_i \triangle^2(x_1...x_{L-1}) \prod_{i=1}^{L-1} (\lambda_1 - x_i)(\lambda_2 - x_i) = \frac{\prod\limits_{j=1}^{L-1} (j!\sqrt{2\pi})}{\lambda_2 - \lambda_1} \left | \begin{matrix} h_{L-1}(\lambda_1) & h_{L-1}(\lambda_2) \\ h_{L}(\lambda_1) & h_{L}(\lambda_2) \end{matrix} \right |.
\end{split}
\end{equation}
The resulting expression for $J_1$ reads
\begin{equation}
\label{eqq}
\begin{split}
J_1 \sim & \prod_{j=1}^{L-1} (j!\sqrt{2\pi}) \frac{e^{-N(L-1)R_*}}{N^{L^2/2}} \int\limits_0^{R_*} dt ~{}_0F_1 \Big (1;-\eta t \Big ) \int\limits_{0}^\infty ds e^{-N\mathcal{L}_2(s)}{}_0 F_1 \Big (1,-\xi s \Big ) \\
& \times \frac{1}{t-s} \left | \begin{matrix} h_{L-1} \Big (\sqrt{N}(s-R_*) \Big ) & h_{L-1}\Big (\sqrt{N}(t-R_*) \Big ) \\ h_{L} \Big (\sqrt{N}(s-R_*) \Big ) & h_{L}\Big (\sqrt{N}(t-R_*) \Big ) \end{matrix} \right |
\end{split}
\end{equation}
and contains two terms due to the determinant. The saddle-point approximation of the $s$ integral alone is performed around the corresponding saddle point $s = R_* + N^{-1/2} \sigma$ and with the use of Rodrigues formula for the Hermite monic polynomials:
\begin{align*}
& \int_0^\infty ds e^{-N\mathcal{L}_2 (s)} \frac{h_{n}\left (\sqrt{N}(s-R_*) \right ) {}_0 F_1 \Big (1, -\xi s \Big )}{t-s} \sim \frac{e^{-NR_*}}{\sqrt{N}} \int\limits_{-\infty}^\infty d\sigma e^{-\frac{\sigma^2}{2}} h_n(\sigma) \frac{{}_0 F_1 \Big (1, -\xi \left ( R_* + \frac{\sigma}{\sqrt{N}} \right ) \Big )}{t-R_*-\frac{\sigma}{\sqrt{N}}} = \nonumber\\
& = \frac{e^{-NR_*}}{\sqrt{N}} \int\limits_{-\infty}^\infty d\sigma e^{-\frac{\sigma^2}{2}} \frac{d^n}{d\sigma^n} \left [ \frac{{}_0 F_1 \Big (1, -\xi \left ( R_* + \frac{\sigma}{\sqrt{N}} \right ) \Big )}{t-R_*-\frac{\sigma}{\sqrt{N}}} \right ] \sim \frac{\sqrt{2\pi} (-1)^{n}e^{-NR_*}}{\sqrt{N}^{n+1}} \frac{d^n}{dR^n} \left [ \frac{{}_0 F_1 \Big (1, -\xi R_* \Big )}{t-R_*} \right ].
\end{align*}
To identify the leading contribution to $J_1$, it is safe to approximate the $t$-dependent Hermite polynomials by its highest power i.e. $\pi_n(\sqrt{N}(t-R_*)) \sim \sqrt{N}^n (t-R_*)^n$. By counting the powers of $N$, leading order contribution to $J_1$ is proportional to $\pi_{L-1} \Big (\sqrt{N}(s-R_*) \Big ) \pi_{L} \Big (\sqrt{N}(t-R_*) \Big )$ and the formula \eqref{eqq} reads
 \begin{align*}
J_1 & \sim (-1)^{L-1} \prod_{j=0}^{L-1} j! \frac{e^{-NLR_*} \sqrt{2\pi}^{L}}{N^{L^2/2}} \int\limits_0^{R_*} dt ~{}_0F_1 \Big (1;-\eta t \Big ) \frac{d^{L-1}}{dR^{L-1}} \left [ \frac{{}_0 F_1 \Big (1, -\xi R_* \Big )}{t-R_*} \right ] (t-R_*)^L .
 \end{align*}
The remaining integral can be represented as a sum
 \begin{equation*}
 \begin{split}
 & \int\limits_0^{R_*} dt ~{}_0F_1 \Big (1;-\eta t \Big ) \frac{d^{L-1}}{dR^{L-1}} \left [ \frac{{}_0 F_1 \Big (1, -\xi R_* \Big )}{t-R_*} \right ] (t-R_*)^L = \\
 & = (-1)^{L-1} (L-1)! \sum_{k=0}^{L-1} \frac{(-\xi)^k}{k!^2} {}_0F_1 \Big (1+k;-\xi R_* \Big ) \int\limits_0^{R_*} dt (t-R_*)^k {}_0F_1 \Big (1;-\eta t \Big ),
 \end{split}
 \end{equation*}
which can be vrified by using the identities $\frac{d^k}{dR^k} {}_0 F_1 \Big (1, -\xi R_* \Big ) = \frac{\xi^k}{k!} {}_0 F_1 \Big (1+k, -\xi R_* \Big )$ and $\frac{d^{L-1-k}}{dR^{L-1-k}} \frac{1}{t-R_*} = (-1)^{L-1-k} \frac{(L-1-k)!}{(t-R_*)^{L-k}}$.
Finally the formula for $J_1$ takes the form
\begin{equation}
\label{J1fin}
J_1 \sim = \prod_{j=0}^{L-1} j! \frac{\sqrt{2\pi}^{L} (L-1)! e^{-NLR_*}}{N^{L^2/2}}  \sum_{k=0}^{L-1} \frac{(-\xi)^k}{k!^2} {}_0F_1 \Big (1+k;-\xi R_* \Big ) \int\limits_0^{R_*} dt (t-1+R)^k {}_0F_1 \Big (1;-\eta t \Big ).
\end{equation}
\subsubsection{Asymptotics of $J_2$}
The procedure of approximating the term $J_2$ in equation \eqref{kappa} is very similar to just discussed. First, applying the saddle-point method to $t_i$ integrals and using formula \eqref{hermident} for the Hermite monic polynomials gives
\begin{equation}
\begin{split}
& J_2 \sim \prod_{j=1}^{L-1} (j!\sqrt{2\pi}) \frac{e^{-N(L-1)R_*}}{N^{L^2/2}} \int\limits_0^{R_*} ds~ {}_0 F_1 (1;-\eta s) {}_0 F_1 (1;-\xi s) \int\limits_{0}^\infty dt e^{-N\mathcal{L}_2(t)}  \frac{1}{t-s} \nonumber \\
& \times\Big [ h_{L-1} \Big (\sqrt{N}(s-R_*) \Big ) h_{L}\Big (\sqrt{N}(t-R_*) \Big ) - h_{L-1}\Big (\sqrt{N}(t-R_*)\Big ) h_{L}\Big (\sqrt{N}(s-R_*)\Big ) \Big ].
\end{split}
\end{equation}
The $t$ integral is in turn found through the saddle-point approximation by expanding $t=R_* + N^{-1/2} \xi$:
\begin{align*}
& \int\limits_{0}^\infty dt e^{-N\mathcal{L}_2(t)}  \frac{1}{t-s} h_{n}\left (\sqrt{N}(t-1+R) \right ) \sim \frac{(-1)^n e^{-NR_*} \sqrt{2\pi} }{\sqrt{N}^{n+1}} \frac{n!}{(R_* - s)^{n+1}} .
\end{align*}
Similarly as in the computation of $J_1$, the leading order term is proportional to $\pi_{L-1}\Big (\sqrt{N}(t-R_*)\Big ) \pi_{L}\Big (\sqrt{N}(s-R_*)\Big )$ which gives
\begin{equation}
\label{J2fin}
J_2 \sim \prod\limits_{j=0}^{L-1} j! \frac{ \sqrt{2\pi}^{L} (L-1)! e^{-NLR_*} }{N^{L^2/2}}\int\limits_0^{R_*} ds~ {}_0 F_1 (1;-\eta s) {}_0 F_1 (1;-\xi s).
\end{equation}
We substitute both formulas \eqref{J1fin} and \eqref{J2fin} for $J_1$ and $J_2$ respectively into the equation \eqref{kappa}, and use the asymptotics of normalization constant $\tilde{\mathcal{N}}_L$ given by formula \eqref{nlapprox}. After reintroducing the renormalized kernel defined in the equation \eqref{kdef} and applying straightforward manipulations, we arrive at
 \begin{equation*}
 \begin{split}
& \frac{1}{N} K_N\left (x = \frac{\xi}{N},y = \frac{\eta}{N} \right ) \sim \\
& \sim \left ( \frac{\eta}{\xi} \right )^L \int\limits_0^{R_*} ds~ {}_0 F_1 (1;-\eta s) \left ( {}_0 F_1 (1;-\xi s) - \sum_{k=0}^{L-1} \frac{(-\xi)^k}{k!^2} {}_0F_1 \Big (1+k;-\xi R_* \Big ) (s-R_*)^k \right ).
 \end{split}
 \end{equation*}
Next we define a 'macroscopic' kernel
\begin{equation*}
\kappa_L(\alpha,\beta) \equiv \lim_{N\to \infty} \frac{1}{NR_*} K_N \left (x = \frac{\alpha}{NR_*},y = \frac{\beta}{NR_*}\right ),
\end{equation*}
which is found by changing the variables as $\xi = \frac{\alpha}{R_*}, \eta = \frac{\beta}{R_*}$, $s = R_*\tau$ and introducing Bessel functions via the identity ${}_0 F_1 (1+k;-\alpha) = \frac{k!}{\alpha^{k/2}} J_k(2\sqrt{\alpha})$:
\begin{equation*}
\kappa_L(\alpha,\beta) = \left ( \frac{\beta}{\alpha} \right )^L \int\limits_0^{1} d\tau~ J_0(2\sqrt{\beta \tau} ) \left ( J_0(2\sqrt{\alpha\tau}) - \sum_{k=0}^{L-1}  (1-\tau)^k \frac{\alpha^{k/2}}{k!} J_k(2\sqrt{\alpha}) \right ).
\end{equation*}
 One then arrives at the final expression \eqref{mainkernelformula} by using the following Proposition for $a = 2\sqrt{\alpha}$ and $b = 2\sqrt{\beta}$.
\end{proof}
\begin{prop}
\label{propkernel}
For $a,b>0$  the following identity holds true
\begin{equation}
\label{propkernelform}
\begin{split}
&\left ( \frac{b}{a} \right )^{2L} \int\limits_0^{1} d\tau~ J_0(b\sqrt{\tau} ) \left ( J_0(a\sqrt{\tau}) - \sum_{k=0}^{L-1}  (1-\tau)^k \frac{(a/2)^{k}}{k!} J_k(a) \right ) = \left ( \frac{b}{a} \right )^{L} \int_0^1 d\tau J_L(a\sqrt{\tau}) J_L(b\sqrt{\tau}).
\end{split}
\end{equation}
\end{prop}
\begin{proof} Proof of this formula is provided in the Appendix \ref{proof}.
\end{proof}
\section{Acknowledgements}
The research at King's College London was supported by  EPSRC grant  EP/N009436/1 "The many faces of random characteristic polynomials".
JG acknowledges partial support from the National Science Centre, Poland under an agreement 2015/19/N/ST1/00878.

\appendix
\section{Equivalence between equation \eqref{DL1} and Rem. \ref{DInverseCPRelation}}
\label{AppL1}
Equality of the formula \eqref{DL1} and the $L=1$ case of Remark \ref{DInverseCPRelation} follows by the following Proposition.
\begin{prop}
\label{L1Equivalence}
For $p>0$ the following relation holds true
\begin{equation}
\label{L1EquivalenceFormula}
\begin{split}
& \frac{(-1)^{N-1}}{2\pi i} \int\limits_{0}^{\infty}du\frac{e^{-u}u}{p+u}
\int\limits_{C}\frac{e^{-v}{}_0 F_1 \left (1;vu \right )dv}{\left(v-|z|^2\right)^N}
\int\limits_0^{\infty} \frac{\left(q+|z|^2\right)^N}{q+v}
e^{-q}dq = \\
& = \frac{1}{(N-1)!} e^{|z|^2} \int\limits_0^\infty \frac{e^{-pt}}{t(1+t)} e^{-\frac{t|z|^2}{1+t}} \left ( \frac{t}{1+t} \right )^{N} \left ( \Gamma(N+1,|z|^2) - \frac{t}{1+t} |z|^2 \Gamma(N,|z|^2) \right ).
\end{split}
\end{equation}
\end{prop}
\begin{proof}
We follow the same approach as in the proof of Prop. \ref{L0Equivalence}. We first bring the l.h.s. to the following form, see Prop. \ref{mynorm},
\begin{equation}
\label{Eu5}
\text{l.h.s.} = \int_0^{\infty} \frac{dt}{(1+t)^2} e^{-tp}\, \mathcal{G}^{(2)}_N\left(\rho,\tau\right),
\end{equation}
where $\tau = \frac{t}{1+t}, \rho = |z|^2$ and (cf. \ref{G1amynorm})
\begin{equation}
\label{G1a}
\mathcal{G}^{(2)}_N\left(\rho,\tau\right) = \frac{(-1)^{N-1}}{(N-1)!}\int_0^{\infty}\,dq\,(q+\rho)^Ne^{-q} \frac{d^{N-1}}{d\rho^{N-1}} \left [ \frac{e^{-\rho\tau}(1+\rho(1-\tau))}{q+\rho} \right ].
\end{equation}
Now we use the Leibniz formula:
\begin{equation}
\frac{d^{N-1}}{d\rho^{N-1}}\left[ e^{-\rho\,\tau}\frac{1+\rho(1-\tau)}{q+\rho}\right] = (N-1)!\sum_{l=0}^{N-1}\frac{(-1)^l}{(N-1-l)!}\frac{1}{(q+\rho)^{l+1}} \frac{d^{N-1-l}}{d\rho^{N-1-l}}\Big [ e^{-\rho\tau}(1+\rho(1-\tau))\Big ],
\end{equation}
which allows to perform the integration over $q$ in formula (\ref{G1a}) explicitly:
\[
\int\limits_0^{\infty}\,dq\,(q+\rho)^{N-l-1}e^{-q}=(N-l-1)!\sum_{m=0}^{N-l-1} \frac{\rho^m}{m!},
\]
and get, after introducing $k=N-l-1$, that
\begin{equation}\label{G1b}
\mathcal{G}^{(2)}_N\left(\rho,\tau\right)=\sum_{k=0}^{N-1}(-1)^k\left(\sum_{m=0}^{k} \frac{\rho^m}{m!}\right) \frac{d^{k}}{d\rho^{k}}\Big [ e^{-\rho\tau}(1+\rho(1-\tau))\Big ].
\end{equation}
By using the Leibniz formula again
\begin{equation}
\frac{d^{k}}{d\rho^{k}}\left(e^{-\rho\tau}(1+\rho(1-\tau))\right)=(-1)^k e^{-\rho\tau}\left[\tau^k(1+\rho(1-\tau))-k\tau^{k-1}(1-\tau)\right],
\end{equation}
we find
\begin{equation}\label{G1c}
\mathcal{G}^{(2)}_N\left(\rho,\tau\right)=e^{-\rho\tau}\sum_{k=0}^{N-1}\left[-\tau^{k+1}\rho+\tau^k(k+1+\rho)-k\tau^{k-1}\right]\left(\sum_{m=0}^{k} \frac{\rho^m}{m!}\right),
\end{equation}
which can be expanded in powers of $\tau$ as
\begin{equation}\label{G1d}
\begin{split}
  \mathcal{G}^{(2)}_N\left(\rho,\tau\right)= e^{-\rho\tau}& \left\{ -\tau^{N}\rho \sum_{m=0}^{N-1} \frac{\rho^m}{m!}+\tau^{N-1}\left(N \sum_{m=0}^{N-1} \frac{\rho^m}{m!}+\frac{\rho^N}{(N-1)!}\right) + \right. \\
 & \left. +\sum_{l=2}^N \tau^{N-l}\left[-(N-l+1)\sum_{m=0}^{N-l+1} \frac{\rho^m}{m!}+(N-l+1+\rho)\sum_{m=0}^{N-l} \frac{\rho^m}{m!}-\rho\sum_{m=0}^{N-l-1} \frac{\rho^m}{m!}\right]\right\}.
 \end{split}
\end{equation}
It is easy to see that coefficients of $\tau^k$ for $k\leq N-2$ are zero, hence we arrive at
\begin{equation}\label{G1f}
\mathcal{G}^{(2)}_N\left(\rho,\tau\right)=e^{-\rho\tau}\tau^{N-1}\left[-\tau\rho \sum_{m=0}^{N-1} \frac{\rho^m}{m!}+ N \sum_{m=0}^{N} \frac{\rho^m}{m!}\right].
\end{equation}
We recall the definition of the incomplete $\Gamma-$function $\Gamma(n,x) = (n-1)! e^{-x} \sum\limits_{m=0}^{n-1}  \frac{x^m}{m!}$ and find
\begin{equation}\label{G1ffin}
 \mathcal{G}^{(2)}_N\left(\rho,\tau\right)= \frac{e^{\rho(1-\tau)}\tau^{N-1}}{(N-1)!}\Big [ -\tau \rho \Gamma(N,\rho) + \Gamma(N+1,\rho) \Big ].
\end{equation}
Finally, by recalling that $\tau=\frac{t}{1+t}$ and $\rho=|z|^2$ we see that $\mathcal{G}^{(2)}_N\left(|z|^2,\frac{t}{1+t}\right)$ substituted to (\ref{Eu5}) coincides with (\ref{L1EquivalenceFormula}).
\end{proof}

\section{Alternative proof of Proposition \ref{InverseCPAveragingTheorem} bypassing use of (bi-)orthogonal polynomials}
We provide an alternative proof of Proposition \ref{InverseCPAveragingTheorem} without using the bi-orthogonal polynomials, relying instead upon the following
\label{ICPalternative}
\begin{prop}
\label{alternativeprop}
The integral
\begin{equation} \label{mainingred}
F^{(L)}_K(\epsilon_1,...,\epsilon_K) = \int\limits_0^{\infty} ... \int\limits_0^{\infty} \,
\triangle(y_1,...,y_K)\, \det{\Big ({}_0 F_1 \left(1;\epsilon_i y_j\right)\Big )}_{i,j=1}^K \prod_{i=1}^K y_i^L e^{-y_i} dy_i,
\end{equation}
can be rewritten as
\begin{equation} \label{mainingredfin}
F^{(L)}_K(\epsilon_1,...,\epsilon_K)= \frac{K! e^{\sum_{i=1}^K \epsilon_i}}{(-1)^{KL + \frac{K(K-1)}{2}} } \,
\det\left(\begin{array}{cccc}
\pi_L(-\epsilon_1)& \pi_L(-\epsilon_2)& \ldots & \pi_L(-\epsilon_K) \\
\pi_{L+1}(-\epsilon_1)& \pi_{L+1}(-\epsilon_2)& \ldots &  \pi_{L+1}(-\epsilon_K)\\
\vdots & \vdots & \ddots & \vdots  \\
 \pi_{L+K-1}(-\epsilon_1) & \pi_{L+K-1}(-\epsilon_2) & \ldots &  \pi_{L+K-1}(-\epsilon_K)
 \end{array}\right),
\end{equation}
where $\pi_n(x) = (-1)^n n! L_n(x)$ denote the Laguerre monic polynomials.
\end{prop}
\begin{proof}
The proof uses the permutation symmetry of the integral \eqref{mainingred}
\begin{equation*} \label{mainingredA}
\begin{split}
F^{(l)}_K(\epsilon_1,\ldots,\epsilon_K) & = K!\int\limits_0^{\infty} ... \int\limits_0^{\infty} \,
\triangle(y_1,...,y_K)\,  \prod_{i=1}^K {}_0 F_1 \left(1;\epsilon_i y_i\right) y_i^L e^{-y_i} dy_i = \\
& = K!\sum_{\alpha_1,\ldots,\alpha_K}(-1)^{\sigma_{\alpha}}
\int\limits_0^{\infty} dy_1 e^{-y_1} y_1^{L+\alpha_1} {}_0 F_1 \left(1;\epsilon_1 y_1\right) ... \int\limits_0^{\infty}dy_K e^{-y_K} y_K^{L+\alpha_K} {}_0 F_1 \left(1;\epsilon_K y_K\right),
\end{split}
\end{equation*}
where summation goes over all $K!$ permutations $\alpha_1,\ldots,\alpha_K$ of the set $0,1,\ldots,K-1$. We apply the identity \eqref{Integral} in each of the above integrals and obtain the formula \eqref{mainingredfin}.
\end{proof}

To arrive at the formula \eqref{InverseCPAveragingForm} for the average of inverse characteristic polynomial, we start from an explicit average computed with respect to the joint probability density given by the equation \eqref{Density1}:
\begin{equation*}
\E \left ( \prod_{i=1}^N \frac{1}{y-x_i}\right ) = \frac{1}{\mathcal{N}_L} \int\limits_0^\infty ... \int\limits_0^\infty \triangle(x_1,...,x_N) \det{\Big ({}_0 F_1 \left(1;\omega_i x_j\right)\Big )}_{i,j=1}^N  \prod_{i=1}^N \frac{1}{y-x_i} x_i^L e^{-x_i} dx_i.
\end{equation*}
where $\mathcal{N}_L = F_N^{(L)}(\omega_1,...,\omega_N)$ is expressed in terms of formula \eqref{mainingred}.
We utilize the permutation symmetry to simplify one of the determinants and apply the identity \eqref{Xident1L} to the product $\prod_{i=1}^N \frac{1}{y-x_i}$ which results in
\begin{equation*}
\begin{split}
\E \left ( \prod_{i=1}^N \frac{1}{y-x_i}\right ) = & \frac{N!}{\mathcal{N}_L} \sum_{i=1}^N \int\limits_0^\infty dx_i \frac{1}{y-x_i} {}_0 F_1 \left(1;\omega_i x_i\right)  x_i^L e^{-x_i}\\
& \times \int\limits_0^\infty ... \int\limits_0^\infty \frac{\triangle(x_1,...,x_N)}{\prod\limits_{\substack{j=1\\j\neq i}}^N (x_i - x_j)} \prod_{\substack{j=1\\ j\neq i}}^N {}_0 F_1 \left(1;\omega_j x_j\right)  x_j^L e^{-x_j} dx_j.
\end{split}
\end{equation*}
Since $\frac{\triangle(x_1,\ldots,x_N)}{\prod_{j\ne i}(x_i-x_j)}=(-1)^{N+i}\Delta(x_1,\ldots,x_{i-1},x_{i+1},\ldots, x_N)$, the $N-1$ integrals of previous formula are expressible in terms of \eqref{mainingred}, we find with renaming $x_i =u$
\begin{equation}
\label{efin}
\begin{split}
\E \left ( \prod_{i=1}^N \frac{1}{y-x_i}\right ) = & N \sum_{i=1}^N (-1)^{N+i} \int\limits_0^\infty du \frac{{}_0 F_1 \left(1;\omega_i u\right)  u^L e^{-u}}{y-u} \frac{F_{N-1}^{(L)}(\omega_1,...\omega_{i-1},\omega_{i+1},...,\omega_N)}{F_N^{(L)}(\omega_1,...,\omega_N)}.
\end{split}
\end{equation}
By using the Proposition \ref{alternativeprop}, the ratio of $F$'s is given by
\begin{equation}
\begin{split}
& \frac{F_{N-1}^{(L)}(\omega_1,...\omega_{i-1},\omega_{i+1},...,\omega_N)}{F_N^{(L)}(\omega_1,...,\omega_N)} = \\
& = \frac{(-1)^{L+N-1}e^{-\omega_i}}{N} \frac{\left|\begin{array}{cccccc} \pi_L(-\omega_1)& \ldots & \pi_L(-\omega_{i-1})& \pi_L(-\omega_{i+1})& \ldots & \pi_L(-\omega_{N})  \\
\pi_{L+1}(-\omega_1)& \ldots & \pi_{L+1}(-\omega_{i-1}) &  \pi_{L+1}(-\omega_{i+1})&  \ldots & \pi_{L+1}(-\omega_{N})\\
\hdots & \ddots & \hdots & \hdots & \ddots & \hdots \\
\pi_{L+N-2}(-\omega_1) & \ldots &  \pi_{L+N-2}(-\omega_{i-1}) & \pi_{L+N-2}(-\omega_{i-1})& \ldots &  \pi_{L+N-2}(-\omega_N)
 \end{array}\right|}
 {\left|\begin{array}{cccc} \pi_L(-\omega_1)& \pi_L(-\epsilon_2)& \ldots & \pi_L(-\omega_N) \\ \pi_{L+1}(-\omega_1)& \pi_{L+1}(-\omega_2)& \ldots &  \pi_{L+1}(-\omega_N)\\ \ldots & \ldots & \ldots & \ldots  \\ \pi_{L+N-1}(-\omega_1) & \pi_{L+N-1}(-\omega_2) & \ldots &  \pi_{L+N-1}(-\omega_N)
 \end{array}\right|},
\end{split}
\end{equation}
which, when substituted back into the equation \eqref{efin} recovers exactly the formula \eqref{GeneralFormula3} without the use of bi-orthogonal structures.

\section{Proof of Proposition \ref{propkernel}}
\label{proof}
In this Appendix we provide a proof of Proposition \ref{propkernel}.
\begin{proof}
The left- and right-hand side of the formula \eqref{propkernelform} is denoted by $\kappa_L$ and $\kappa_L'$ respectively:
\begin{equation}
\begin{split}
\kappa_L & = \left ( \frac{b}{a} \right )^{2L} \int\limits_0^{1} d\tau~ J_0(b\sqrt{\tau} ) \left ( J_0(a\sqrt{\tau}) - \sum_{k=0}^{L-1}  (1-\tau)^k \frac{(a/2)^{k}}{k!} J_k(a) \right ), \\
\kappa_L' & = \left ( \frac{b}{a} \right )^{L} \int_0^1 d\tau J_L(a\sqrt{\tau}) J_L(b\sqrt{\tau}).
\end{split}
\end{equation}
We compute
\begin{equation*}
\int_0^{1} d\tau (1-\tau)^k J_0(b\sqrt{\tau}) = \frac{2k!}{b} \left ( \frac{2}{b} \right )^k J_{k+1} (b)
\end{equation*}
so that $\kappa_L$ reads
\begin{equation*}
\kappa_L = \left ( \frac{b}{a} \right )^{2L} \left [ \int\limits_0^{1} d\tau J_0 \left (a\sqrt{\tau} \right ) J_0 \left (b\sqrt{\tau} \right ) - \frac{2}{b} \sum_{k=0}^{L-1} \left ( \frac{a}{b} \right )^{k} J_k(a) J_{k+1}(b) \right ].
\end{equation*}
We prove the equality $\kappa_L = \kappa_L'$ by induction. For $L=0$, the second sum term vanishes and we trivially obtain an equality
\begin{align}
\kappa_0 = \int\limits_0^{1} d\tau J_0 \left (a\sqrt{\tau} \right ) J_0 \left (b\sqrt{\tau} \right ) = \kappa_0',
\end{align}
which also agrees with Theorem 3.1 given in  Ref. \cite{ForresterLiu}. We assume the validity of $\kappa_L = \left ( \frac{b}{a} \right )^{L} \int_0^1 d\tau J_L(a\sqrt{\tau}) J_L(b\sqrt{\tau})$ and consider $\kappa_{L+1}$:
\begin{align*}
 \kappa_{L+1} & = \frac{b^2}{a^2} \left ( \frac{b}{a} \right )^{2L} \left [ \int\limits_0^{1} ds J_0 \left (a \sqrt{s} \right ) J_0 \left (b \sqrt{s} \right ) - \frac{2}{b} \sum_{k=0}^{L-1} \left ( \frac{a}{b} \right )^{k} J_k(a) J_{k+1}(b) - \frac{2}{b} \left ( \frac{a}{b} \right )^{L} J_L(a) J_{L+1}(b)  \right ] = \\
& = \frac{b^2}{a^2} \kappa_L -  \frac{2}{b} \left ( \frac{b}{a} \right )^{L+2} J_L(a) J_{L+1}(b).
\end{align*}
On the $\kappa_L$ term we use an identity
\begin{equation}
\int_0^1 ds J_L(a\sqrt{s}) J_L(b\sqrt{s}) = 2 \frac{b J_{L-1}(b) J_L(a) - a J_{L-1}(a) J_L(b)}{a^2-b^2}
\end{equation}
and utilize recurrence relations $ J_{L-1}(x) = \frac{2L}{x} J_L(x) - J_{L+1}(x)$ to eliminate $J_{L-1}$ resulting in
\begin{align*}
\kappa_{L+1} = & 2 \left ( \frac{b}{a} \right )^{L+2} \left [\frac{b J_{L-1}(b) J_L(a) - a J_{L-1}(a) J_L(b)}{a^2-b^2} - \frac{b - a^2/b}{b^2-a^2} J_L(a) J_{L+1}(b) \right ] = \\
 = & \left ( \frac{b}{a} \right )^{L+1} \int_0^1 ds J_{L+1}(a\sqrt{s}) J_{L+1}(b\sqrt{s}) = \kappa_{L+1}'.
 \end{align*}
\end{proof}

\end{document}